%% file: main.tex
\documentclass[preprint,11pt,authoryear]{elsarticle}

\usepackage[a4paper,margin=2.5cm]{geometry}
\usepackage{setspace}
\usepackage{amsmath,amssymb,amsfonts}
\usepackage[ruled]{algorithm2e}
\usepackage{amsthm}
\usepackage{graphicx}
\usepackage{makecell}
\usepackage{subcaption}
\usepackage{booktabs}
\usepackage{placeins}
\usepackage{url}

\newtheorem{definition}{Definition}
\newtheorem{proposition}{Proposition}
\newtheorem{theorem}{Theorem}
\newtheorem{lemma}{Lemma}

\biboptions{authoryear}
\begin{document}

\begin{frontmatter}

\title{Where to Split and When to Charge: Optimal Route Construction from Customer Permutations in Electric Vehicle Routing}

\author[inst1]{Leon Stjepan Uroi\'c\corref{cor1}}
\ead{leon-stjepan.uroic@fer.hr}

\author[inst1]{Marko {\DJ}urasevi\'c}
\ead{marko.durasevic@fer.hr}

\affiliation[inst1]{
  organization={University of Zagreb Faculty of Electrical Engineering and Computing},
  addressline={Unska 3},
  city={Zagreb},
  postcode={10000},
  country={Croatia}
}

\begin{abstract}
Permutation-based metaheuristics are widely used for electric vehicle routing, where candidate solutions are represented as ordered sequences of customers. Such sequences, however, do not directly define feasible vehicle routes: they must be decoded by choosing where to split the permutation into routes and where to insert charging-station visits, subject to cargo capacity and battery constraints. These decisions are inherently interdependent, since each return to the depot both separates consecutive routes and restores the vehicle battery. This paper formalizes the task as the Fixed-Permutation Splitting and Charging Problem and proposes an exact forward labeling algorithm that constructs a minimum-distance feasible decoding of a fixed customer permutation using dynamic programming with dominance pruning. We further derive restricted variants representing increasingly simplified decoding strategies: first separating route splitting from charging-station insertion, and then additionally limiting each inter-customer segment to at most one charging-station visit. Computational experiments on benchmark and randomly generated instances, including comparisons with heuristic decoders from the literature, confirm that the exact decoder remains tractable in practice and reveal a clear hierarchy among decoding strategies. The most restrictive variant achieves runtimes close to those of heuristic decoders while delivering substantially higher decoding success rates and better solution quality. Less restrictive variants further improve quality and robustness at the cost of additional runtime. The exact joint decoder provides the optimal reference for each fixed permutation, clarifying the trade-offs introduced by common decoding simplifications.
\end{abstract}

\begin{keyword}
Routing \sep
electric vehicle routing \sep
dynamic programming \sep
permutation decoding \sep
metaheuristics
\end{keyword}

\end{frontmatter}

\input{introduction}
\input{related_work}
\input{problem_statement}
\input{method}
\input{results}
\input{conclusion}

\section{Acknowledgments}
The work of doctoral student Leon Stjepan Uroić has been fully supported by the "Young researchers' career development project – training of doctoral students" of the Croatian Science Foundation DOK-2025-02-8664. The authors would like to thank Ivan Milinović, with whom this line of research was originally initiated, for the early collaborative work that laid the groundwork for this paper. Leon Stjepan Uroić would also like to thank Ana Sović Kržić for her continued guidance and support, and Liljana Puškar for carefully reading the manuscript and providing valuable feedback.

\section*{Declaration of generative AI and AI-assisted technologies in the manuscript preparation process}
During the preparation of this work the authors used Claude and ChatGPT in order to improve the language, readability, and flow of the manuscript. After using this tool, the authors reviewed and edited the content as needed and take full responsibility for the content of the published article.

\bibliography{references}

\input{appendix}

\end{document}

%% file: introduction.tex
\section{Introduction}
Governments worldwide are tightening emission standards and introducing low-emission zones in urban centers \citep{jonidi2021urban}. Meanwhile, e-commerce is growing rapidly, and with it the demand for last-mile delivery in cities \citep{bhatti2020commerce}. As a consequence, an increasing number of delivery companies are incorporating electric vehicles into their fleets \citep{juan2016electric}. While electric vehicles eliminate tailpipe emissions and reduce fuel costs, they also introduce new planning challenges: battery capacity is limited, recharging takes significantly longer than refueling, and charging infrastructure is unevenly distributed and less developed than refueling infrastructure \citep{oliveira2017sustainable, juan2016electric}. These constraints fundamentally change the nature of the underlying routing problem.

Formally, the task of constructing delivery routes is known as the \textit{Vehicle Routing Problem} (\textbf{VRP}). The VRP and many of its variants have been extensively studied over the past several decades \citep{braekers2016vehicle, eksioglu2009vehicle}. However, the methods and algorithms developed for the VRP are designed for vehicles with internal combustion engines, where refueling is nearly instantaneous, gas stations are ubiquitous, and driving range is not a major concern. None of these assumptions hold for electric vehicles, and as a result a new branch of routing problems, known as the \textit{Electric Vehicle Routing Problem} (\textbf{EVRP}), has received growing attention \citep{kucukoglu2021electric, erdelic2019survey}. Beyond the usual capacity and routing constraints, solving the EVRP also requires deciding where and when to recharge, which charging stations to visit, and in what sequence, while respecting battery limits that couple decisions across the entire route. This combinatorial explosion in the decision space is what distinguishes the EVRP from the classical VRP and renders existing methods insufficient without adaptation.

While a variety of approaches have been explored for the EVRP, metaheuristics are among the most widely used and have reported strong performance~\citep{kucukoglu2021electric, erdelic2019survey}. To cope with the aforementioned complexity, metaheuristic approaches often decompose the problem into multiple levels. A common strategy is to first partition customers into routes and then invoke a separate procedure to insert charging station visits in order to satisfy the battery constraint \citep{feng_bilevel_2024, jia2021bilevel, jia2022confidence}. This decomposition has become common enough that a dedicated line of research, known as the \textit{Fixed Route Vehicle Charging Problem} (\textbf{FRVCP}), has emerged around this charging-insertion subproblem \citep{montoya2016electric}.

However, many metaheuristic algorithms commonly applied to the EVRP naturally operate on permutation-based representations. For instance, the use of permutations as solution representations has been extensively studied for genetic algorithms, both in the general context and in the context of VRP variants, and a large library of crossover and mutation operators exists in the literature \citep{pavai2016survey, larranaga1999genetic, vidal2022hybrid}. The same holds for other metaheuristics \citep{9504893, tian1999application}. Consequently, there is a need for algorithms that can produce a valid EVRP solution from a customer permutation; we refer to this problem as the \textit{Fixed-Permutation Splitting and Charging Problem} (\textbf{FPSCP}). The task of such an algorithm is to strategically insert depot and charging station visits into a customer permutation so that the resulting solution has minimum total distance while satisfying all constraints. 

Although the FRVCP has received considerable attention \citep{montoya2016multi, roberti2016electric, hiermann2016electric, schiffer2018adaptive, deschenes2020fixed, kullman2021frvcpy, froger2019improved}, the broader FPSCP, which additionally requires splitting customers into routes, has received limited attention and remains largely unexplored. Yet this problem arises naturally in permutation-based metaheuristics, where a customer ordering must be converted into a complete feasible EVRP solution before it can be evaluated. An optimal FPSCP solver also makes it possible to quantify the effect of common heuristic simplifications, such as separating route splitting from charging insertion or limiting the number of charging-station visits between consecutive customers.

To address these gaps, \textbf{this paper makes the following contributions}:
\begin{itemize}
\item formulating the Fixed-Permutation Splitting and Charging Problem, which arises in permutation-based EVRP methods and jointly captures route splitting and charging-station insertion from a fixed customer ordering,
\item proposing an optimal forward-labeling algorithm for FPSCP and proving that it returns an optimal decoding for any fixed feasible customer permutation,
\item deriving restricted variants of the proposed framework: a split-then-charge decoder and a further restriction limiting charging-station insertion to at most one station per inter-customer segment, enabling direct comparison across decoding strategies,
\item quantifying the impact of the aforementioned decoding simplifications through computational experiments comparing the proposed variants against heuristic decoders in terms of solution quality, solve rate, and runtime.
\end{itemize}

The remainder of this paper is organized as follows. Section~\ref{sec:related-work} reviews related work on electric vehicle routing, fixed-route charging problems, and decoding procedures used in permutation-based metaheuristics. Section~\ref{sec:problem-statement} formally defines the Fixed-Permutation Splitting and Charging Problem. Section~\ref{sec:algorithm} presents the proposed forward labeling algorithm and proves its optimality. Section~\ref{sec:results} reports computational experiments on benchmark and randomly generated instances, including comparisons with restricted and heuristic decoders. Finally, Section~\ref{sec:conclusion} concludes the paper and discusses directions for future work.

%% file: related_work.tex
\section{Related Work}
\label{sec:related-work}

The \textit{Vehicle Routing Problem}, first introduced by \citet{dantzig1959truck}, is one of the most extensively studied combinatorial optimization problems~\citep{braekers2016vehicle, eksioglu2009vehicle}. As a generalization of the \textit{Traveling Salesman Problem} (\textbf{TSP}), it is NP-hard~\citep{lenstra1981complexity}, and decades of research have produced a rich landscape of solution methods ranging from exact algorithms and mathematical programming to construction heuristics and population-based metaheuristics~\citep{braekers2016vehicle}. The operations research community has further extended the problem into many practically relevant variants that accommodate real-world constraints such as time windows, heterogeneous fleets, and stochastic conditions~\citep{konstantakopoulos2022vehicle}.

With the introduction of electric vehicles, a new routing problem called the \textit{Electric Vehicle Routing Problem} emerged~\citep{kucukoglu2021electric, erdelic2019survey}. Since electric vehicles have a shorter range than conventional vehicles, recharging takes significantly longer than refueling, and charging infrastructure is less developed, recharging decisions cannot be ignored during route planning. As a result, the EVRP is considerably harder than the classical VRP, since a solver must additionally decide when to recharge, which charging stations to visit, and how to ensure that the vehicle never runs out of battery while en route.

A range of approaches have been proposed for the EVRP. On the exact side, \citet{tahami2020exact} introduced a branch-and-cut method able to solve instances of moderate size to optimality. For dynamic variants, in which instance parameters may change while vehicles are en route, hyper-heuristic methods have been explored as a way to automatically adapt search strategies to shifting conditions~\citep{durasevic2024automated}. More recently, deep learning has emerged as an active direction, with several works framing the EVRP as a sequential decision making problem and training neural policies to construct routes~\citep{lin2021deep, basso2022dynamic, tang2023energy}. Despite this growing diversity, metaheuristics remain the dominant paradigm in the literature and continue to report the strongest results on the WCCI-2020 benchmark~\citep{mavrovouniotis2020benchmark}.

Many successful metaheuristic approaches rely on a bilevel decomposition of the EVRP. In these methods, an upper level determines the customer service order or route partitioning, while a lower level inserts charging-station visits to restore energy feasibility. For example, \citet{jia2021bilevel} formulate the problem in this way and combine an order-first split-second ant-colony procedure with a heuristic FRVCP solver. In a later paper, they improve this framework through a confidence-based selection mechanism and a simplified charging-station enumeration procedure~\citep{jia2022confidence}. \citet{hien2023greedy} also adopt a bilevel structure, combining route construction with a greedy charging-insertion strategy. \citet{feng_bilevel_2024} use the same decomposition, but replace the upper-level solver with Hybrid Genetic Search~\citep{vidal2022hybrid} and introduce a more advanced lower-level charging heuristic. Other approaches, such as that proposed by \citet{woller2020grasp}, similarly rely on a repair phase that first restores route feasibility and then inserts charging stations, while \citet{rodriguez2024new} instead treat routing and charging decisions within a unified hyper-heuristic framework. Overall, these works show that charging insertion is a central component of high-performing EVRP solvers, whether handled explicitly as a lower-level subproblem or embedded within a more general search procedure.

The bilevel decomposition is so prevalent that the lower-level \textit{Fixed Route Vehicle Charging Problem} has become a research topic in its own right. The FRVCP assumes that the customer visit sequence is fixed, including both the order of customers and their assignment to routes, and focuses on inserting charging-station visits so as to ensure energy feasibility while minimizing travel cost. Since the FRVCP is NP-hard~\citep{montoya2016electric}, many EVRP solvers rely on heuristic FRVCP procedures, including several of the methods discussed above. At the same time, a number of exact approaches have been proposed by exploiting the close connection between the FRVCP and the \textit{Constrained Shortest Path} (\textbf{CSP}) problem. \citet{montoya2016multi} formulate the problem on an auxiliary graph and solve it optimally using the pulse algorithm~\citep{lozano2013exact}. \citet{roberti2016electric} consider both full and partial recharging with time windows and propose a forward labeling algorithm based on precomputed recharging paths between consecutive customers. \citet{hiermann2016electric} embed FRVCP-related procedures both in a bidirectional labeling algorithm for pricing and in a simplified forward labeling routine within their \textit{Adaptive Large Neighbourhood Search} (\textbf{ALNS}) heuristic, where at most one charging station is allowed between consecutive customers. \citet{schiffer2018adaptive} use a dynamic programming approach with labels that maintain feasible arrival-time corridors under partial recharging, and also derive a constant-time penalty evaluation scheme for neighborhood search. \citet{deschenes2020fixed} extend the FRVCP to jointly optimize charging decisions and driving speeds under nonlinear charging and consumption functions. \citet{froger2019improved} propose improved formulations and algorithmic components for this broader setting, and \citet{kullman2021frvcpy} provide an open-source FRVCP solver that facilitates benchmarking and reproducibility.

Despite this growing body of work on the FRVCP, most existing studies assume that routes are already given and therefore do not address the broader problem that arises in permutation-based metaheuristics. In that setting, transforming a customer permutation into a complete feasible EVRP solution requires simultaneously deciding where to split the permutation into routes and where to insert charging-station visits. These two decisions are inherently coupled, since the depot both partitions routes and charges the vehicle. The resulting \textit{Fixed-Permutation Splitting and Charging Problem} is therefore not merely a straightforward extension of the FRVCP, but a distinct decoding problem that has received comparatively little attention in the literature. An optimal FPSCP solver is also valuable from an analytical perspective, since it makes it possible to quantify the effect of common simplifications, such as separating route splitting from charging insertion or restricting the decoder to at most one charging station between consecutive customers.

This paper addresses this gap by proposing an optimal algorithm for the FPSCP that integrates route construction and charging decisions within a unified forward labeling framework. The proposed method can be viewed as an extension of the work of \citet{roberti2016electric} from the FRVCP to the FPSCP setting. A further important difference is that their method assumes that the energy required to travel between any pair of charging stations is smaller than the vehicle's battery capacity, whereas we make no such simplifying assumption. The resulting algorithm efficiently generates optimal solutions for fixed customer permutations while remaining computationally tractable in practice.

%% file: problem_statement.tex
\section{Problem Statement}
\label{sec:problem-statement}

The \textit{Fixed-Permutation Splitting and Charging Problem} arises as a subproblem in permutation-based metaheuristics for the \textit {Electric Vehicle Routing Problem}. Given a fixed ordering of customers, the FPSCP seeks to transform this permutation into a complete, feasible EVRP solution of minimum total distance by simultaneously deciding (i) where to split the permutation into individual routes, each starting and ending at the depot, and (ii) where to insert charging station visits along the resulting routes. A feasible solution must satisfy the following constraints:
\begin{itemize}
    \item The relative order of customers in the input permutation must be preserved.
    \item Every route must start and end at the depot.
    \item Each customer must be visited exactly once.
    \item The total demand of customers on a given route cannot exceed the vehicle's maximum cargo capacity.
    \item Vehicles must maintain sufficient battery energy throughout the route.
\end{itemize}

The problem is defined on a weighted, fully connected graph \(G=(V,E)\). The vertex set \( V=\{\delta\} \cup V_c \cup V_f \) includes a single depot node \(\delta\), a set of customer nodes \(V_c=\{c_1, c_2, \ldots, c_n\}\), and a set of charging stations \( V_f \). Each edge \((i,j) \in E\) is weighted by the Euclidean distance \( d_{ij} \), and traversing it consumes \( h \cdot d_{ij} \) units of battery energy, where \(h \in \mathbb{R}^+\) denotes the battery consumption rate. Each customer \(i \in V_c\) has a demand \(q_i\). All vehicles are identical, sharing the same cargo capacity \(Q\) and battery capacity \(B\), and are assumed to depart from the depot fully charged and fully loaded.

The input to the FPSCP is a permutation \( \pi = (\pi_1, \pi_2, \ldots, \pi_n) \) of the customer set \(V_c\), which fixes the relative order in which customers must be served. A solution is a sequence of nodes
\[
    S = (s_0, s_1, \ldots, s_L), \quad s_0 = s_L = \delta, \quad s_\ell \in V \text{ for } \ell = 0, \ldots, L
\]
obtained by inserting depot visits and charging station visits into \(\pi\). Removing all non-customer nodes from \(S\) must yield exactly \(\pi\); this is the order-preservation requirement that distinguishes the FPSCP from the general EVRP. The depot visits partition \(S\) into a set of routes \(\mathcal{R}(S) = \{R_1, R_2, \ldots, R_K\}\), where each route \(R_k\) is a maximal subsequence of \(S\) starting and ending at \(\delta\) and containing no intermediate depot visit.

For a route \(R = (r_0, r_1, \ldots, r_T)\) with \(r_0 = r_T = \delta\), the total customer demand served on that route is
\[
    q(R) = \sum_{\substack{t=0,\ldots,T \\ r_t \in V_c}} q_{r_t}.
\]
The battery state \(b_t\) upon arrival at node \(r_t\) denotes the amount of battery energy consumed since the most recent visit to either the depot or a charging station. It is initialized as
\[
    b_0 = 0,
\]
and evolves according to
\[
    b_{t+1} =
    \begin{cases}
        h d_{r_t r_{t+1}}, 
        & \text{if } r_t \in V_f \cup \{\delta\}, \\[2mm]
        b_t + h d_{r_t r_{t+1}}, 
        & \text{if } r_t \in V_c,
    \end{cases}
    \qquad t = 0,\ldots,T-1.
\]
Thus, leaving a charging station or the depot resets the accumulated battery consumption to zero before the next trip, whereas leaving a customer does not. A solution \(S\) is feasible if every route \(R \in \mathcal{R}(S)\) satisfies
\begin{align}
    q(R) &\leq Q \qquad \forall R \in \mathcal{R}(S), \label{eq:capacity} \\
    b_t &\leq B \qquad \text{for } t = 0,\ldots,T. \label{eq:battery}
\end{align}
Condition~\eqref{eq:capacity} enforces the cargo capacity limit, while condition~\eqref{eq:battery} ensures that the vehicle never runs out of battery along the route. The objective of the FPSCP is to find a feasible solution that minimizes the total travel distance
\begin{equation}
    min \hspace{1em} f(S) = \sum_{\ell=0}^{L-1} d_{s_\ell s_{\ell+1}}. \label{eq:objective}
\end{equation}
subject to the aforementioned constraints.

%% file: method.tex
\section{Optimal FPSCP Forward Labeling Algorithm}
\label{sec:algorithm}

In this section, we present an optimal forward labeling algorithm for the FPSCP, which we denote by \textbf{FP-FLA} (\textit{Fixed-Permutation Forward Labeling Algorithm}). The algorithm processes the customer permutation from left to right and maintains, after each customer, a Pareto front of non-dominated partial solutions. Each partial solution is extended by considering three types of transitions: direct travel to the next customer, detours through charging stations, and detours that include a depot visit. The complete procedure is given in Algorithm~\ref{alg:FP-FLA}, while the remainder of the section explains the label structure, the extension rules, and the proof of optimality.

\begin{algorithm}[htbp]
\caption{FP-FLA: Optimal FPSCP Forward Labeling Algorithm}
\label{alg:FP-FLA}
\SetAlgoLined
\KwIn{A permutation of customer nodes \( \pi = (\pi_1, \pi_2, \ldots, \pi_n) \), cargo capacity $Q$, battery capacity $B$}
\KwOut{A minimum-distance feasible decoding of \( \pi \), if one exists}

Initialize $front_0$ with the single label $(d^0 = 0, q^0 = 0, b^0 = 0)$\;

Compute the distance matrix $\mathbf{F}$ for all charging stations, including the depot\;

$\pi' \leftarrow (\delta, \pi_1, \dots, \pi_n, \delta)$\;

\For{$i \leftarrow 1$ \KwTo $|\pi'|-1$}{
    Initialize $front_i \leftarrow \emptyset$\;
    \ForEach{$L^{i-1}_j \in front_{i-1}$}{
    
        \tcp{Option 1: Direct extension}
        $L^{i}_{\text{direct}} \gets$ extend $L^{i-1}_j$ directly to $\pi'_i$ using equation~\eqref{eq:direct}\;
        \If{$q^i_{\text{direct}} \le Q \land b^i_{\text{direct}} \le B$}{
            Add $L^{i}_{\text{direct}}$ to $front_i$\;
        }

        \ForEach{$f_{\text{in}} \in V_f \cup \{\delta\}$}{
            \ForEach{$f_{\text{out}} \in V_f \cup \{\delta\}$}{
                \tcp{Option 2: Charging-station detour}
                $L^{i}_{\text{charge}} \gets$ extend $L^{i-1}_j$ via a charging detour from $f_{\text{in}}$ to $f_{\text{out}}$ using equation~\eqref{eq:charge}\;
                \If{$b^{i-1} + h d_{\pi_{i-1}, f_in} \le B \land q^i_{\text{charge}} \le Q \land b^i_{\text{charge}} \le B$}{
                    Add $L^{i}_{\text{charge}}$ to $front_i$\;
                }

                \tcp{Option 3: Depot detour}
                $L^{i}_{\text{depot}} \gets$ extend $L^{i-1}_j$ via a depot detour from $f_{\text{in}}$ to $f_{\text{out}}$ using equation~\eqref{eq:depot}\;
                \If{$b^{i-1} + h d_{\pi_{i-1}, f_in} \le B \land q^i_{\text{depot}} \le Q \land b^i_{\text{depot}} \le B$}{
                    Add $L^{i}_{\text{depot}}$ to $front_i$\;
                }
            }
        }
    }

    Remove dominated labels from $front_i$\;
}
\If{$front_{|\pi'|-1} = \emptyset$}{
    \Return{No feasible solution}\;
}
\Return{the minimum-distance label in $front_{|\pi'|-1}$ and its reconstructed solution}\;
\end{algorithm}

The algorithm takes a permutation \( \pi = (\pi_1, \pi_2, \ldots, \pi_n) \) of customer nodes \(V_c\) as input. A partial solution up to \(\pi_i\) is represented by the label
\[
L^i = (d^i, q^i, b^i).
\]
Here, \(d^i\) denotes the total distance traveled up to the \(i\)-th customer, including all depot and charging-station visits. The term \(q^i\) represents the cargo used since the most recent depot visit, while \(b^i\) denotes the battery consumed since the last visit to either a charging station or the depot. Additionally, each label must store backtracking information in order to reconstruct the final solution. This information is omitted here for clarity. Based on the triple \((d^i, q^i, b^i)\), a natural dominance relation arises between labels.

\begin{definition}[Pareto Domination]
Let
\[
L^i_j = (d^i_j, q^i_j, b^i_j) \quad \text{and} \quad L^i_k = (d^i_k, q^i_k, b^i_k)
\]
be two distinct labels corresponding to partial solutions up to \(\pi_i\). We say that \(L^i_j\) is \emph{dominated} by \(L^i_k\), denoted \(L^i_j \prec L^i_k\), if and only if
\[
d^i_k \le d^i_j, \quad
q^i_k \le q^i_j, \quad
b^i_k \le b^i_j,
\]
and at least one of these inequalities is strict.
\end{definition}

The proposed method is a forward labeling algorithm that processes the customer permutation from the start to the end. At step \(i\), it starts from the set of non-dominated partial solutions up to \(\pi_{i-1}\) and extends them to \(\pi_i\). As illustrated in Figure~\ref{fig:jump-cases}, each label can be extended in three different ways.

\begin{figure}[htbp]
    \centering
    \includegraphics[width=0.70\textwidth]{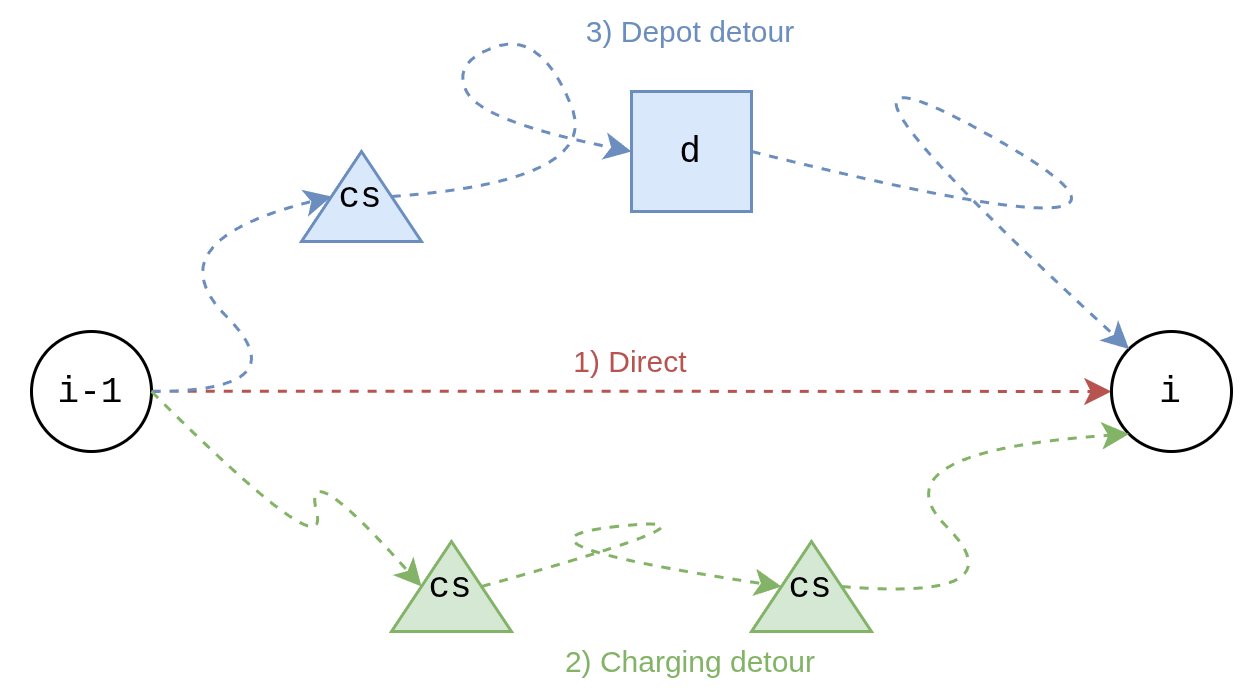}
    \caption{The three possible extensions of a partial solution from the $(i-1)$-th customer to the $i$-th customer.}
    \label{fig:jump-cases}
\end{figure}

The first option is to travel directly from customer \(\pi_{i-1}\) to customer \(\pi_i\). The resulting label \(L^i_{\text{direct}}\) is defined as
\begin{equation}
\begin{aligned}
    d^i_{\text{direct}} &= d^{i-1} + d_{\pi_{i-1}, \pi_i}, \\
    q^i_{\text{direct}} &= q^{i-1} + q_{\pi_i}, \\
    b^i_{\text{direct}} &= b^{i-1} + h d_{\pi_{i-1}, \pi_i}.
\end{aligned}
\label{eq:direct}
\end{equation}

The second option is to visit one or more charging stations before proceeding to the next customer. For this purpose, we temporarily treat the depot as a charging station, while ignoring the fact that a depot visit also resets the vehicle's cargo state. Such extensions are handled separately later through depot detours.

A charging detour is characterized by the first charging station visited after \(\pi_{i-1}\), denoted by \(f_{\text{in}}\), and the last charging station visited before \(\pi_i\), denoted by \(f_{\text{out}}\). The choice of \(f_{\text{in}}\) is constrained by the remaining battery at \(\pi_{i-1}\), since not every charging station is reachable from the current partial solution. Similarly, not every customer \(\pi_i\) is reachable from every charging station \(f_{\text{out}}\), even with a fully charged battery. Once \(f_{\text{in}}\) and \(f_{\text{out}}\) are fixed, however, the intermediate charging stations depend only on this pair and not on the surrounding route. We now formalize this observation.

\begin{definition}[Optimal Sequence of Intermediate Charging Stations]
Let \(f_{\text{in}}, f_1, f_2, \ldots, f_K, f_{\text{out}}\) be a sequence of charging stations visited between customer nodes \(\pi_{i-1}\) and \(\pi_i\). The sequence of intermediate charging stations \(f_1, f_2, \ldots, f_K\) is said to be \emph{optimal} for a given pair \((f_{\text{in}}, f_{\text{out}})\) if, for every partial solution with label \(L^{i-1}\), the resulting label \(L^i\) dominates every label obtained by any alternative choice of intermediate charging stations between \(f_{\text{in}}\) and \(f_{\text{out}}\).
\label{def:optimal-intermediate-cs}
\end{definition}

We can now state and prove this property formally.

\begin{proposition}
Let \(f_{\text{in}}, f_1, f_2, \ldots, f_K, f_{\text{out}}\) be a sequence of charging stations visited between customer nodes \(\pi_{i-1}\) and \(\pi_i\). The optimal sequence of intermediate charging stations \(f_1, f_2, \ldots, f_K\) depends only on the first charging station \(f_{\text{in}}\) and the last charging station \(f_{\text{out}}\), and is independent of the remainder of the route.
\label{prop:cs}
\end{proposition}

\begin{proof}
The cargo component at customer node \(\pi_i\) does not depend on the choice of intermediate charging stations between \(\pi_{i-1}\) and \(\pi_i\), since visiting charging stations does not affect cargo. Likewise, because the vehicle is fully recharged at the last charging station \(f_{\text{out}}\), the battery component \(b^i\) upon arrival at \(\pi_i\) is also independent of the choice of intermediate charging stations. Therefore, by Definition~\ref{def:optimal-intermediate-cs}, an optimal sequence of intermediate charging stations must minimize only the total traveled distance.

For a fixed partial solution with label \(L^{i-1}\), the total distance traveled to reach customer \(\pi_i\) can be written as
\begin{equation}
    d^i ={} d^{i-1} + d_{\pi_{i-1}, f_{\text{in}}} + d_{f_{\text{in}}, f_1}
          + \sum_{k=2}^{K} d_{f_{k-1}, f_k} + d_{f_K, f_{\text{out}}}
          + d_{f_{\text{out}}, \pi_i}.
\end{equation}
Since \(d^{i-1}\), \(\pi_{i-1}\), \(\pi_i\), \(f_{\text{in}}\), and \(f_{\text{out}}\) are fixed, minimizing \(d^i\) reduces to minimizing
\[
d_{f_{\text{in}}, f_1} + \sum_{k=2}^{K} d_{f_{k-1}, f_k} + d_{f_K, f_{\text{out}}},
\]
which depends only on \(f_{\text{in}}\) and \(f_{\text{out}}\). This proves independence from the route following \(f_{\text{out}}\).

It remains to show independence from the route preceding \(f_{\text{in}}\). Since \(f_{\text{in}}\) is a charging station, every feasible sequence of intermediate charging stations starts from \(f_{\text{in}}\) with a fully charged battery. As noted above, the cargo component is unaffected by the intermediate charging sequence. Hence, the feasibility and cost of the intermediate sequence do not depend on the route preceding \(f_{\text{in}}\). Therefore, the optimal sequence depends only on \(f_{\text{in}}\) and \(f_{\text{out}}\).
\end{proof}

\begin{proposition}
Given \(f_{\text{in}}\) and \(f_{\text{out}}\), the optimal sequence of intermediate charging stations can be computed as the shortest path from \(f_{\text{in}}\) to \(f_{\text{out}}\) in the graph
\[
G = (V_f \cup \{\delta\}, E), \quad
E = \{\{u,v\} \mid u,v \in V_f \cup \{\delta\},\; h d_{u,v} \le B\},
\]
without any additional constraints.
\label{prop:shortest-path}
\end{proposition}

\begin{proof}
By Proposition~\ref{prop:cs}, the optimal sequence of intermediate charging stations is precisely the feasible sequence minimizing
\[
d_{f_{\text{in}}, f_1} + \sum_{k=2}^{K} d_{f_{k-1}, f_k} + d_{f_K, f_{\text{out}}}.
\]

We model this as a path problem in the graph \(G\), whose vertices are the charging stations together with the depot. An edge \(\{u,v\}\) is present exactly when a fully charged vehicle can travel from \(u\) to \(v\), that is, when \(h d_{u,v} \le B\). Since the vehicle is fully recharged at every charging station, every path in \(G\) is energy-feasible. Therefore, finding the optimal sequence of intermediate charging stations is equivalent to finding a shortest path from \(f_{\text{in}}\) to \(f_{\text{out}}\) in \(G\).
\end{proof}

This allows us to precompute a distance matrix \(\mathbf{F}\), whose entries store the shortest-path distances between all pairs of charging stations, including the depot. These distances can be computed using the Floyd--Warshall algorithm~\citep{floyd1962algorithm} in \(\mathcal{O}(|V_f|^3)\) time.

We can now extend a label \(L^{i-1}\) to customer node \(\pi_i\) via a charging detour as
\begin{equation}
\begin{aligned}
    d^i_{\text{charge}}(f_{\text{in}}, f_{\text{out}})
    &= d^{i-1} + d_{\pi_{i-1}, f_{\text{in}}}
     + \mathbf{F}_{f_{\text{in}}, f_{\text{out}}}
     + d_{f_{\text{out}}, \pi_i}, \\
    q^i_{\text{charge}}(f_{\text{in}}, f_{\text{out}})
    &= q^{i-1} + q_{\pi_i}, \\
    b^i_{\text{charge}}(f_{\text{in}}, f_{\text{out}})
    &= h d_{f_{\text{out}}, \pi_i}.
\end{aligned}
\label{eq:charge}
\end{equation}
There are \((|V_f|+1)^2\) possible choices for the pair \((f_{\text{in}}, f_{\text{out}})\), where the “+1” accounts for treating the depot as a charging station.

Finally, the third way to extend a label \(L^{i-1}\) is via a depot detour, in which the vehicle may visit multiple charging stations between consecutive customers but must visit the depot. In this case, the label is extended as
\begin{equation}
\begin{aligned}
    d^i_{\text{depot}}(f_{\text{in}}, f_{\text{out}})
    &= d^{i-1} + d_{\pi_{i-1}, f_{\text{in}}}
     + \mathbf{F}_{f_{\text{in}}, \delta}
     + \mathbf{F}_{\delta, f_{\text{out}}}
     + d_{f_{\text{out}}, \pi_i}, \\
    q^i_{\text{depot}}(f_{\text{in}}, f_{\text{out}})
    &= q_{\pi_i}, \\
    b^i_{\text{depot}}(f_{\text{in}}, f_{\text{out}})
    &= h d_{f_{\text{out}}, \pi_i}.
\end{aligned}
\label{eq:depot}
\end{equation}
As in the charging-detour case, there are \((|V_f|+1)^2\) possible choices for the pair \((f_{\text{in}}, f_{\text{out}})\). We now show that this extension is optimal.

\begin{lemma}
Given \(f_{\text{in}}\) and \(f_{\text{out}}\), the optimal sequence of intermediate charging stations among all feasible sequences that visit the depot has total distance
\[
d = \mathbf{F}_{f_{\text{in}}, \delta} + \mathbf{F}_{\delta, f_{\text{out}}},
\]
where \(\mathbf{F} \in \mathbb{R}^{(|V_f|+1)\times(|V_f|+1)}\) is the matrix whose entries \(\mathbf{F}_{i,j}\) store the shortest-path distance between charging stations \(i\) and \(j\), including the depot, in the graph \(G\).
\label{lemma:depot-detour}
\end{lemma}

\begin{proof}
By Proposition~\ref{prop:shortest-path}, the optimal sequence of intermediate charging stations is a shortest path in the graph \(G\).

A shortest path in a graph with no negative cycles never visits the same node more than once, since any repeated node would induce a cycle that could be removed without increasing the path length. Therefore, any shortest path from \(f_{\text{in}}\) to \(f_{\text{out}}\) that visits the depot visits it exactly once. Its total distance can thus be decomposed into the shortest path from \(f_{\text{in}}\) to \(\delta\) and the shortest path from \(\delta\) to \(f_{\text{out}}\), yielding
\[
d = \mathbf{F}_{f_{\text{in}}, \delta} + \mathbf{F}_{\delta, f_{\text{out}}}.
\]
This proves the claim.
\end{proof}

Now that all possible extensions have been defined, the branching factor at each step is \(\mathcal{O}(|V_f|^2)\). Although this may appear prohibitive, most generated labels are either infeasible or dominated and can therefore be discarded. Let
\begin{equation}
    \mathcal{C}_i ={} \{L^i_{\text{direct}}\} \;\cup
    \{ L^i_{\text{charge}}(u,v) \mid u,v \in V_f \cup \{\delta\} \} \;\cup
    \{ L^i_{\text{depot}}(u,v) \mid u,v \in V_f \cup \{\delta\} \}
\end{equation}
denote the set of candidate labels generated at step \(i\). The algorithm then removes all dominated labels from \(\mathcal{C}_i\). This can be done efficiently using an algorithm for computing the maxima of a point set in three dimensions, which runs in \(\mathcal{O}(|\mathcal{C}_i|\log |\mathcal{C}_i|)\) time when implemented with a balanced search tree. The remaining non-dominated labels form the Pareto front \(front_i\), from which the labels for the next step are extended. As shown empirically in the next section, these Pareto fronts remain small across all considered instances, which makes the algorithm efficient in practice. It remains to prove the optimality of the algorithm.

\begin{theorem}[Optimality of FP-FLA]
Given a fixed customer permutation \( \pi = (\pi_1, \pi_2, \ldots, \pi_n) \), Algorithm~\ref{alg:FP-FLA} returns a feasible solution of minimum total distance whenever such a solution exists. Therefore, the algorithm computes an optimal decoding of the given permutation.
\end{theorem}

\begin{proof}
We prove by induction on \(i\) that the Pareto front \(front_i\) maintained by the algorithm contains at least one non-dominated label corresponding to an optimal feasible partial decoding of the prefix \((\pi_1,\dots,\pi_i)\).

\textbf{Base case:} For \(i=0\), the unique label in \(front_0\) represents the empty partial solution at the depot with zero distance, zero cargo, and zero battery consumption. This is trivially optimal.

\textbf{Inductive step:} Assume the claim holds for \(front_{i-1}\). Consider an optimal feasible partial decoding of the prefix \((\pi_1,\dots,\pi_i)\). Its last extension from \(\pi_{i-1}\) to \(\pi_i\) must be of exactly one of the following three types:
\begin{itemize}
    \item direct travel from \(\pi_{i-1}\) to \(\pi_i\),
    \item a charging detour through charging stations, with some entry--exit pair \((f_{\text{in}}, f_{\text{out}})\),
    \item a depot detour, again with some entry--exit pair \((f_{\text{in}}, f_{\text{out}})\).
\end{itemize}
By the inductive hypothesis, \(front_{i-1}\) contains a label corresponding to an optimal partial decoding up to \(\pi_{i-1}\). The algorithm generates all feasible extensions of this label of the three types listed above. By Proposition~\ref{prop:shortest-path}, the charging-detour extension is computed optimally for every pair \((f_{\text{in}}, f_{\text{out}})\), and by Lemma~\ref{lemma:depot-detour}, the same holds for depot detours.

When the depot is treated as a charging station in a charging-detour extension, the cargo component is not reset, as in the equation~\eqref{eq:charge}. Such a label is therefore never better than the corresponding depot-detour label with the same pair \((f_{\text{in}}, f_{\text{out}})\): both labels have the same distance and battery component, while the depot-detour label has no larger cargo component because the depot reset is applied explicitly. Hence, any such charging-detour label is dominated by the corresponding depot-detour label and will be discarded.

By construction, every feasible partial solution that can lead to a globally optimal solution for the first \(i\) customers is generated. Dominated solutions are removed, but by the definition of Pareto dominance, no removed solution could produce a better total distance than a non-dominated solution. Hence, at least one optimal extension remains in \(front_i\).

\textbf{Termination:} After processing all customers, \(front_n\) contains at least one label corresponding to an optimal feasible decoding of the full permutation \((\pi_1,\dots,\pi_n)\). Therefore, Algorithm~\ref{alg:FP-FLA} returns an optimal solution for the given permutation.
\end{proof}

%% file: results.tex
\section{Results}
\label{sec:results}

This section evaluates the proposed method from several perspectives. First, it examines the size of the Pareto fronts maintained during decoding. Second, it compares the solution quality of joint FPSCP decoding with that of split-then-FRVCP approaches, including both multiple- and single-station insertion strategies, as well as three heuristic alternatives. Finally, it analyzes the time required to decode a permutation. All experiments were conducted using the open-source implementation of the proposed algorithms, available at \url{https://github.com/leon3428/FPSCP}.

Customer permutations were generated in two ways. The first uses uniformly random permutations, which serve as a proxy for low-quality solutions that may arise early in the search process. The second uses a stochastic \textit{k-nearest-neighbor} (\textbf{kNN}) procedure, in which the next customer is selected uniformly at random from the $k$ nearest unvisited customers of the current node. These permutations are more structured and better reflect the high-quality orderings typically explored by permutation-based metaheuristics, where consecutive customers tend to be geographically close and only a small number of long jumps remain.

This experimental design evaluates the decoding problem in isolation, decoupling the effects of route splitting and charging-station insertion from the design choices of any particular metaheuristic. This allows decoder-level trade-offs in solution quality, robustness, and runtime to be measured directly and without confounding factors. The two permutation types span a broad range of solution quality, from the unstructured random case to the near-heuristic structured case, providing a controlled setting in which the relative performance of decoding strategies can be assessed cleanly. Since structured permutations are the more relevant setting in practice, the main text focuses on stochastic nearest-neighbor permutations, while results for uniformly random permutations are reported in \ref{app:random-permutations}.

\subsection{Front size}

The practical efficiency of FP-FLA depends strongly on the size of the Pareto front maintained at each decoding step. To examine how the front size varies with instance characteristics, an empirical study was conducted on randomly generated instances in which all node coordinates were sampled uniformly from the unit square $[0,1]^2$ using the baseline parameters listed in Table~\ref{tab:baseline}. Four parameter sweeps were considered: the number of customers, the number of charging stations, the battery capacity, and the cargo capacity, while all remaining parameters were kept fixed. Sweeping the maximum customer demand or the energy-consumption rate would be redundant, since scaling these quantities is equivalent to scaling the cargo capacity and battery capacity, respectively.

\begin{table}[htbp]
    \centering
    \caption{Baseline instance parameters.}
    \label{tab:baseline}
    \begin{tabular}{lr}
        \toprule
        \textbf{Parameter} & \textbf{Value} \\
        \midrule
        Customer count         & 100   \\
        Charging station count & 10    \\
        Battery capacity       & 2.0   \\
        Cargo capacity         & 200.0 \\
        Maximum demand         & 10.0  \\
        Energy consumption     & 1.0   \\
        \bottomrule
    \end{tabular}
\end{table}

For each point in a sweep, 32 random instances were generated, and for each instance 32 stochastic kNN ($k=2$) permutations were constructed. The maximum Pareto front size observed during the execution of the algorithm was then recorded. The resulting trends are shown in Figure~\ref{fig:sweeps}.

\begin{figure*}[htb]
    \centering
    \begin{subfigure}[b]{0.48\linewidth}
        \includegraphics[width=\linewidth]{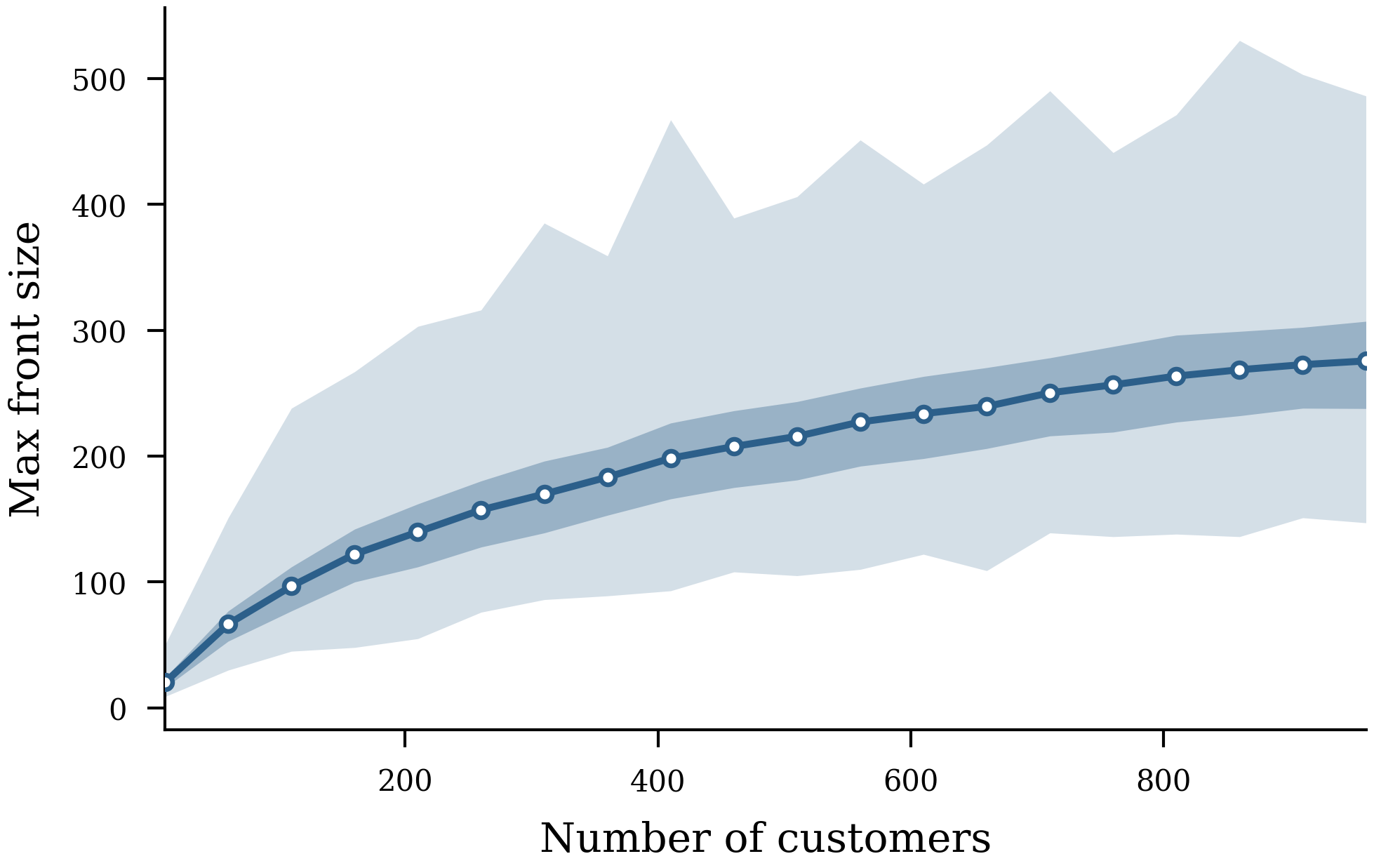}
        \caption{Max front size related to the number of customers}
        \label{fig:customer_cnt_sweep}
    \end{subfigure}
    \hfill
    \begin{subfigure}[b]{0.48\linewidth}
        \includegraphics[width=\linewidth]{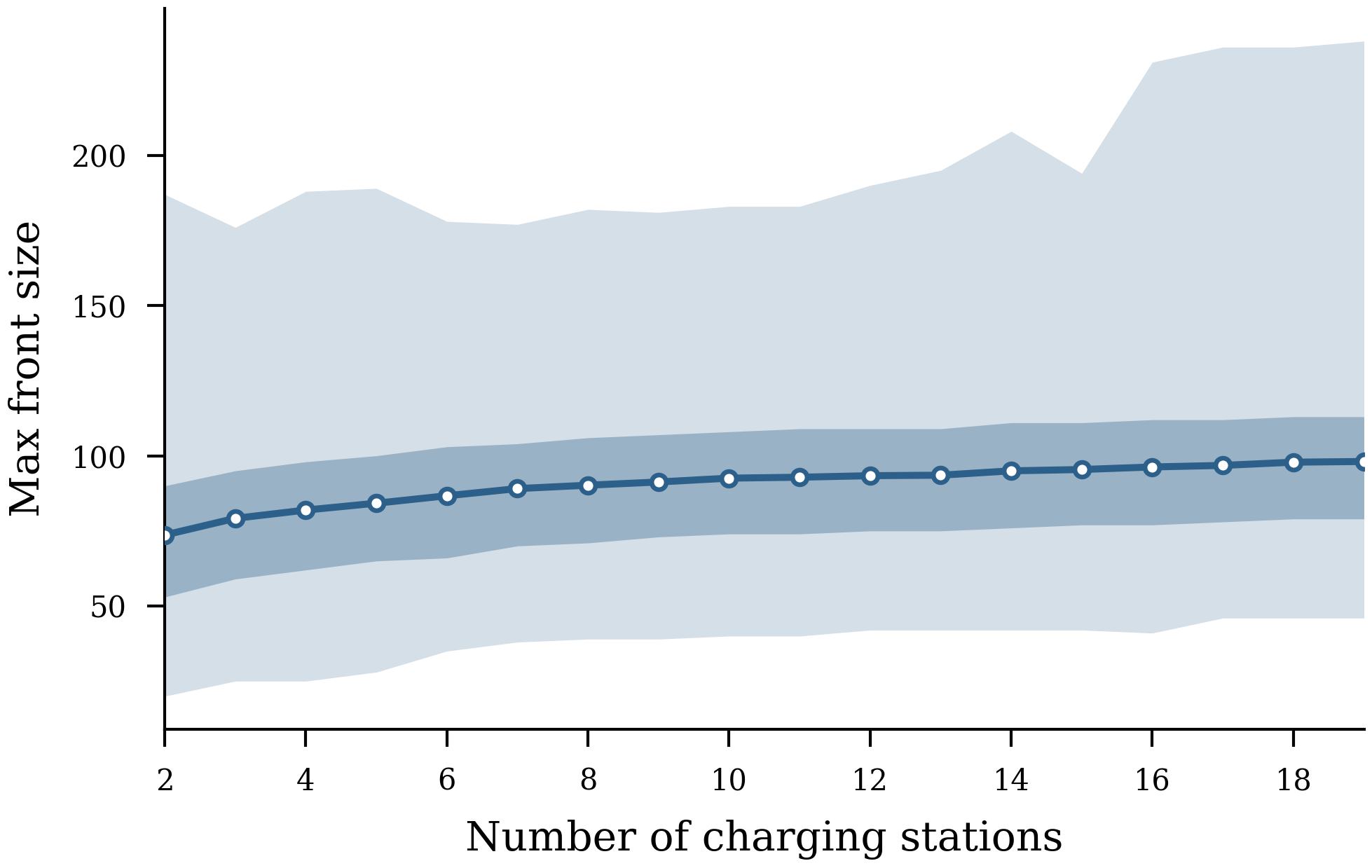}
        \caption{Max front size related to the number of charging stations}
        \label{fig:charging_station_cnt_sweep}
    \end{subfigure}

    \vspace{0.5em}

    \begin{subfigure}[b]{0.48\linewidth}
        \includegraphics[width=\linewidth]{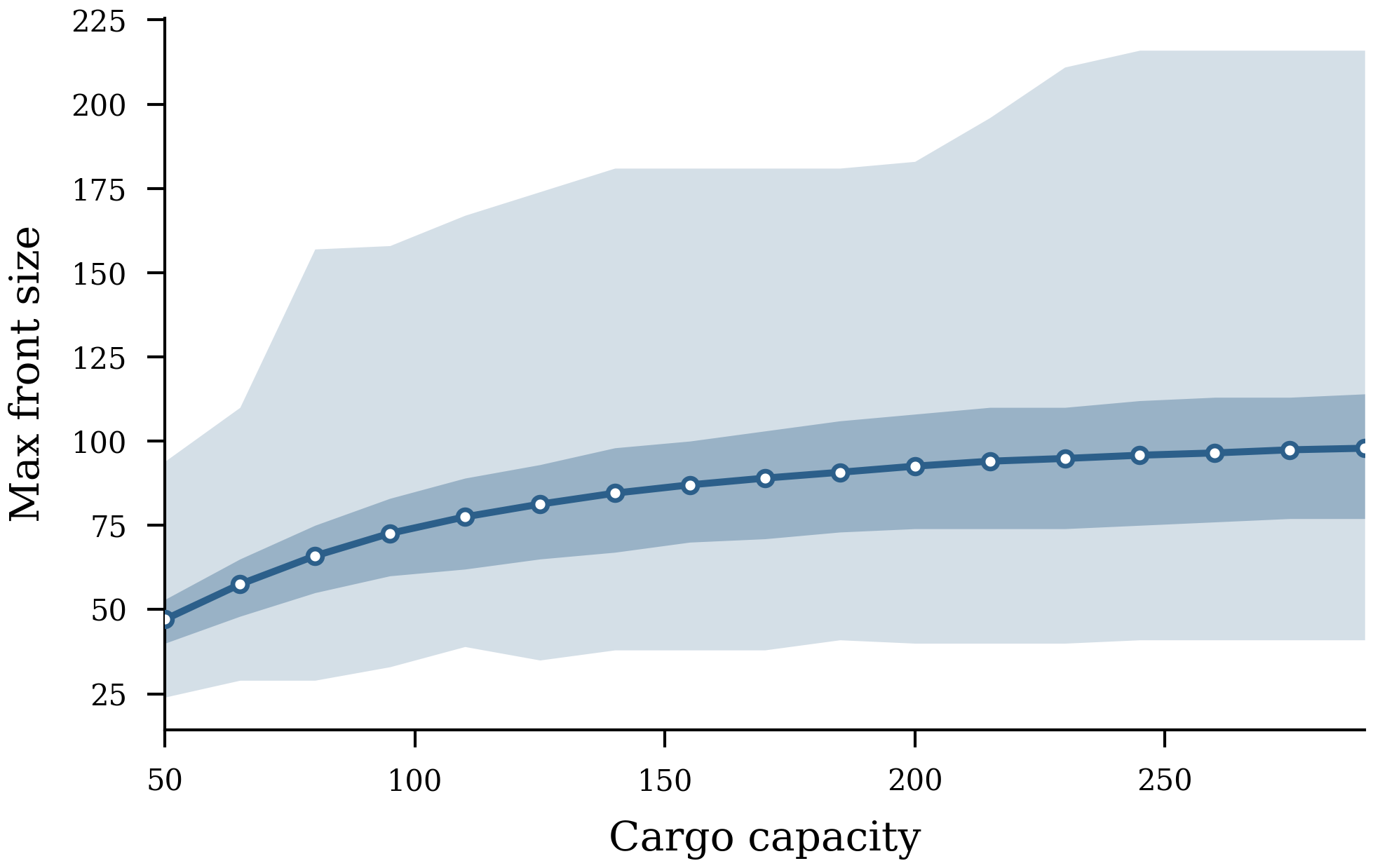}
        \caption{Max front size related to the cargo capacity.}
        \label{fig:cargo_capacity_sweep}
    \end{subfigure}
    \hfill
    \begin{subfigure}[b]{0.48\linewidth}
        \includegraphics[width=\linewidth]{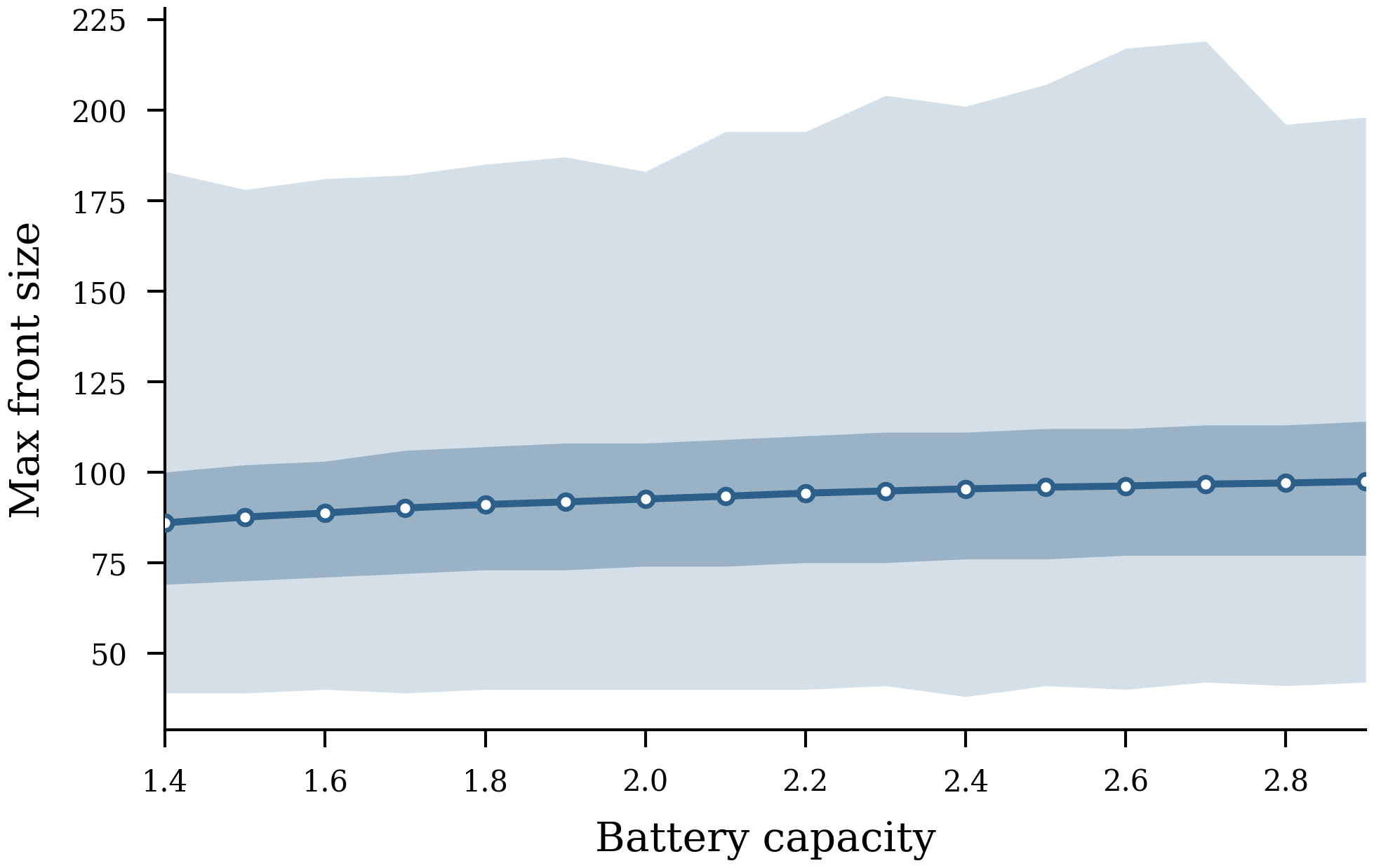}
        \caption{Max front size related to the battery capacity.}
        \label{fig:battery_capacity_sweep}
    \end{subfigure}
    \caption{Plots of the maximum Pareto front size as a function of the swept instance parameter. Each point is based on multiple randomly generated instance configurations and input permutations. The solid line shows the mean, the dark shaded region spans the first to third quartile, and the light shaded region covers the full observed range.}
    \label{fig:sweeps}
\end{figure*} 

The results indicate that the Pareto fronts remain well behaved in practice, with the largest front sizes staying on the order of $10^2$ even for larger instances. Across all four sweeps, the front size initially increases and then gradually levels off as the parameter grows. This behavior is expected: once one parameter becomes sufficiently large, the remaining instance parameters become the dominant limiting factors, so the number of viable non-dominated partial solutions no longer grows substantially. Among the parameters considered, the number of customers has the strongest effect on the front size, whereas the battery capacity has little influence as long as it is large enough for feasible solutions to exist. Although the branching factor grows quadratically with the number of charging stations, the number of customers has a greater effect on the maximal front size, at least for the instance sizes considered here. The number of charging stations affects how many candidate extensions are generated at each step, whereas the number of customers determines how many such steps must be carried out. Empirically, the latter turns out to be the more influential factor.

\subsection{Comparison with split and heuristic decoders}

This subsection compares the proposed optimal FPSCP decoder with six alternative decoding methods. For all FRVCP-based methods, the customer permutation is first partitioned into routes using Vidal's optimal linear-time split algorithm \citep{vidal2016split}.

One of the questions examined in this paper is the performance loss incurred by decoupling route splitting from charging-station insertion, that is, by first splitting the permutation into routes and then solving the resulting FRVCP instances independently. Although this decomposition is common in the literature, its impact on solution quality is rarely quantified. To isolate this effect, FP-FLA was specialized to the FRVCP setting by modifying the label so that it tracks only traveled distance and battery usage, under the assumption that route feasibility with respect to cargo capacity has already been enforced by the preceding split procedure. Under this specialization, the depot-detour extension is no longer needed. In addition, because the input now consists of customer sequences with depot visits already inserted, the battery state is reset whenever the current node is the depot. The resulting algorithm, denoted \textbf{FR-FLA} (\textit{Fixed-Route Forward Labeling Algorithm}), solves the considered FRVCP variant to optimality. Unlike the method of \citet{roberti2016electric}, FR-FLA does not rely on the assumption that every pair of charging stations can be traversed by a fully charged vehicle. Under that assumption, no more than two charging-station visits are ever needed between consecutive customers. In contrast, FR-FLA allows arbitrary sequences of intermediate charging stations while retaining the same branching factor and asymptotic complexity for the considered FRVCP variant. The comparison also includes \texttt{frvcpy} \citep{kullman2021frvcpy}, the open-source implementation of the algorithm of \citet{froger2019improved}. Since that method is also optimal for the considered FRVCP instances, it produces the same solutions as FR-FLA; the comparison is therefore relevant only in terms of runtime.

Another common simplification is to allow at most one charging-station visit between two consecutive customers. To assess the effect of this restriction, FR-FLA was further simplified by removing the double loop over all $(f_{in}, f_{out})$ pairs. This variant is denoted \textbf{SS-FR-FLA} (\textit{Single Station FR-FLA}). Applying the same restriction directly within FP-FLA is not meaningful in the same way, because FP-FLA must simultaneously decide whether to insert a depot visit or a charging-station visit, whereas in FR-FLA the depot visits are already fixed by the split solution. In preliminary experiments, imposing the restriction directly in FP-FLA led to very poor performance, as it severely limited the algorithm's ability to introduce depot returns when needed.

Finally, three heuristic FRVCP decoders were included for comparison: the decoders embedded within BHGA \citep{feng_bilevel_2024}, BACO \citep{jia2021bilevel}, and CACO \citep{jia2022confidence}. These are complete EVRP metaheuristics, each of which incorporates an FRVCP solver as its decoding component. Only this decoding component was extracted and compared here. For convenience, the remainder of the paper refers to each extracted decoder by the name of its parent method. Performance was evaluated on 1000 permutations generated by stochastic kNN sampling with $k=2$ for each WCCI-2020 benchmark instance.

\subsubsection{Percentage of solved permutations}

Methods other than FP-FLA are not guaranteed to return a feasible solution even when one exists. Since FP-FLA is exact, a permutation for which it finds no feasible solution is itself infeasible. The solve rate is therefore evaluated only over permutations that FP-FLA successfully decoded, i.e., those for which a feasible solution is known to exist. The results are reported in Table~\ref{tab:percentage_solved}.

\begin{table}[htbp]
\centering
\small
\setlength{\tabcolsep}{5pt}
\caption{Percentage of feasible stochastic kNN permutations successfully decoded by each method on the WCCI-2020 benchmark instances. Percentages are computed only over permutations for which a feasible solution exists.}
\label{tab:percentage_solved}
\begin{tabular}{lccccc}
\toprule
Instance & \makecell{FR-FLA /\\ FRVCPY} & SS-FR-FLA & BHGA & BACO & CACO \\
\midrule
E-n22-k4    & 100.0\% & 100.0\% & 89.0\% & 97.2\% & 97.2\% \\
E-n23-k3    & 100.0\% & 100.0\% & 86.5\% & 100.0\% & 100.0\% \\
E-n30-k3    & 100.0\% & 100.0\% & 96.7\% & 100.0\% & 100.0\% \\
E-n33-k4    & 100.0\% & 100.0\% & 98.1\% & 100.0\% & 100.0\% \\
E-n51-k5    & 100.0\% & 100.0\% & 87.5\% & 100.0\% & 100.0\% \\
E-n76-k7    & 100.0\% & 100.0\% & 71.9\% & 96.5\% & 96.5\% \\
E-n101-k8   & 100.0\% & 100.0\% & 66.5\% & 95.9\% & 95.9\% \\
X-n143-k7   & 100.0\% & 100.0\% & 99.6\% & 100.0\% & 100.0\% \\
X-n214-k11  & 100.0\% & 100.0\% & 84.3\% & 100.0\% & 100.0\% \\
X-n351-k40  & 100.0\% & 85.8\%  & 26.4\% & 62.1\% & 62.1\% \\
X-n459-k26  & 100.0\% & 100.0\% & 56.6\% & 95.5\% & 95.5\% \\
X-n573-k30  & 100.0\% & 100.0\% & 99.7\% & 100.0\% & 100.0\% \\
X-n685-k75  & 100.0\% & 99.3\%  & 49.1\% & 83.1\% & 83.1\% \\
X-n749-k98  & 100.0\% & 99.7\%  & 49.0\% & 85.1\% & 85.1\% \\
X-n819-k171 & 100.0\% & 99.5\%  & 69.6\% & 90.6\% & 90.6\% \\
X-n916-k207 & 100.0\% & 100.0\% & 97.2\% & 100.0\% & 100.0\% \\
X-n1001-k43 & 100.0\% & 100.0\% & 93.3\% & 100.0\% & 100.0\% \\
\midrule
\textbf{Aggregate}   & \textbf{100.0\%} & \textbf{99.1\%}  & \textbf{77.7\%} & \textbf{94.5\%} & \textbf{94.5\%} \\
\bottomrule
\end{tabular}
\end{table}

For the considered instances, split decoding followed by optimal FRVCP solution, represented by FR-FLA, solved all such permutations. The single-station variant SS-FR-FLA also performed well on many instances and, in several cases, solved nearly all sampled permutations. However, its solved percentage dropped to \(85.8\%\) on instance X-n351-k40, indicating that there are realistic benchmark instances for which restricting the decoder to at most one charging-station visit between consecutive customers can exclude feasible solutions, even when a permutation is split into routes optimally in advance.

The heuristic methods performed noticeably worse. Although these algorithms were designed primarily to produce charging plans for near-optimal permutations rather than to guarantee feasibility whenever possible, the results indicate that they discard a considerable number of feasible permutations. This is particularly relevant in the present setting, where stochastic kNN sampling generates structured permutations with geographically close consecutive customers and only occasional long jumps. Such permutations are representative of the solutions typically encountered during metaheuristic search, where a candidate solution is often already of fairly high quality and differs from a strong solution by only a small number of local modifications, such as 2-opt moves. On the most difficult instance, X-n351-k40, BHGA solved only \(26.4\%\) of feasible permutations, while BACO and CACO each solved \(62.1\%\). Consequently, many potentially promising candidate solutions would never be evaluated further if these decoders were used inside a permutation-based metaheuristic.

\subsubsection{Performance Gap Analysis}

Since the main distinction between the methods lies in the quality of the solutions they produce, we quantify this difference using the percentage performance gap with respect to the solution found by FP-FLA. Because FP-FLA is optimal for a fixed permutation, no lower-cost solution exists for that permutation, making it a natural reference. The performance gap was computed only over the permutations that a given method successfully solved. Since the permutations were sampled, this yielded, for each method--instance pair, an empirical distribution of percentage performance gaps.

For some fraction of the solved permutations, a method matched the FP-FLA solution exactly. This zero-gap percentage is reported in Table~\ref{tab:zero_p90_p95_knn}, together with the 90th and 95th percentiles of the gap distribution. The zero-gap percentage indicates how often a method attains the optimal solution for the sampled permutation, while p90 and p95 characterize the upper tail of the distribution. For FR-FLA/FRVCPY and SS-FR-FLA, the gap distributions were consistently strongly right-skewed, with most permutations concentrated near zero gap and a relatively long tail of worse outcomes. As expected, increasing the number of samples tends to reveal a longer tail. In contrast, the distributions of BHGA, BACO, and CACO were more instance-dependent and ranged from similarly skewed to more symmetric shapes.

\begin{table}[htbp]
\centering
\footnotesize
\setlength{\tabcolsep}{4pt}
\renewcommand{\arraystretch}{1.15}
\caption{Performance-gap summary for stochastic kNN permutations on the WCCI-2020 benchmark instances. Bold values denote the zero-gap percentage, defined as the share of successfully decoded permutations for which the method matches the optimal FP-FLA solution. The p90 and p95 values are the 90th and 95th percentiles of the percentage performance-gap distribution relative to FP-FLA, computed over successfully decoded permutations. Lower p90 and p95 values indicate smaller upper-tail losses.}
\label{tab:zero_p90_p95_knn}
\begin{tabular}{lccccc}
\toprule
Instance & \makecell{FR-FLA/\\FRVCPY} & SS-FR-FLA & BHGA & BACO & CACO \\
\midrule
E-n22-k4 & \makecell[tc]{\textbf{86.5\%}\\[-1pt] \scriptsize p90: 0.03\%\\[-1pt] \scriptsize p95: 0.19\%} & \makecell[tc]{\textbf{77.8\%}\\[-1pt] \scriptsize p90: 0.16\%\\[-1pt] \scriptsize p95: 0.27\%} & \makecell[tc]{\textbf{62.8\%}\\[-1pt] \scriptsize p90: 0.54\%\\[-1pt] \scriptsize p95: 0.95\%} & \makecell[tc]{\textbf{52.4\%}\\[-1pt] \scriptsize p90: 0.57\%\\[-1pt] \scriptsize p95: 1.12\%} & \makecell[tc]{\textbf{59.1\%}\\[-1pt] \scriptsize p90: 0.55\%\\[-1pt] \scriptsize p95: 1.09\%} \\
E-n23-k3 & \makecell[tc]{\textbf{96.4\%}\\[-1pt] \scriptsize p90: 0.00\%\\[-1pt] \scriptsize p95: 0.00\%} & \makecell[tc]{\textbf{75.4\%}\\[-1pt] \scriptsize p90: 0.12\%\\[-1pt] \scriptsize p95: 0.26\%} & \makecell[tc]{\textbf{48.1\%}\\[-1pt] \scriptsize p90: 0.87\%\\[-1pt] \scriptsize p95: 1.30\%} & \makecell[tc]{\textbf{47.3\%}\\[-1pt] \scriptsize p90: 0.68\%\\[-1pt] \scriptsize p95: 0.90\%} & \makecell[tc]{\textbf{43.1\%}\\[-1pt] \scriptsize p90: 0.94\%\\[-1pt] \scriptsize p95: 1.30\%} \\
E-n30-k3 & \makecell[tc]{\textbf{91.7\%}\\[-1pt] \scriptsize p90: 0.00\%\\[-1pt] \scriptsize p95: 0.10\%} & \makecell[tc]{\textbf{88.2\%}\\[-1pt] \scriptsize p90: 0.02\%\\[-1pt] \scriptsize p95: 0.13\%} & \makecell[tc]{\textbf{61.1\%}\\[-1pt] \scriptsize p90: 0.36\%\\[-1pt] \scriptsize p95: 0.71\%} & \makecell[tc]{\textbf{54.9\%}\\[-1pt] \scriptsize p90: 0.51\%\\[-1pt] \scriptsize p95: 0.97\%} & \makecell[tc]{\textbf{57.2\%}\\[-1pt] \scriptsize p90: 0.52\%\\[-1pt] \scriptsize p95: 0.99\%} \\
E-n33-k4 & \makecell[tc]{\textbf{92.5\%}\\[-1pt] \scriptsize p90: 0.00\%\\[-1pt] \scriptsize p95: 0.06\%} & \makecell[tc]{\textbf{90.3\%}\\[-1pt] \scriptsize p90: 0.00\%\\[-1pt] \scriptsize p95: 0.08\%} & \makecell[tc]{\textbf{82.7\%}\\[-1pt] \scriptsize p90: 0.09\%\\[-1pt] \scriptsize p95: 0.24\%} & \makecell[tc]{\textbf{66.2\%}\\[-1pt] \scriptsize p90: 0.32\%\\[-1pt] \scriptsize p95: 0.47\%} & \makecell[tc]{\textbf{70.0\%}\\[-1pt] \scriptsize p90: 0.32\%\\[-1pt] \scriptsize p95: 0.49\%} \\
E-n51-k5 & \makecell[tc]{\textbf{82.6\%}\\[-1pt] \scriptsize p90: 0.09\%\\[-1pt] \scriptsize p95: 0.23\%} & \makecell[tc]{\textbf{68.5\%}\\[-1pt] \scriptsize p90: 0.18\%\\[-1pt] \scriptsize p95: 0.30\%} & \makecell[tc]{\textbf{48.1\%}\\[-1pt] \scriptsize p90: 0.48\%\\[-1pt] \scriptsize p95: 0.67\%} & \makecell[tc]{\textbf{34.9\%}\\[-1pt] \scriptsize p90: 0.52\%\\[-1pt] \scriptsize p95: 0.73\%} & \makecell[tc]{\textbf{39.4\%}\\[-1pt] \scriptsize p90: 0.51\%\\[-1pt] \scriptsize p95: 0.73\%} \\
E-n76-k7 & \makecell[tc]{\textbf{64.2\%}\\[-1pt] \scriptsize p90: 0.37\%\\[-1pt] \scriptsize p95: 0.52\%} & \makecell[tc]{\textbf{49.2\%}\\[-1pt] \scriptsize p90: 0.42\%\\[-1pt] \scriptsize p95: 0.61\%} & \makecell[tc]{\textbf{29.6\%}\\[-1pt] \scriptsize p90: 0.76\%\\[-1pt] \scriptsize p95: 0.96\%} & \makecell[tc]{\textbf{17.3\%}\\[-1pt] \scriptsize p90: 0.93\%\\[-1pt] \scriptsize p95: 1.20\%} & \makecell[tc]{\textbf{20.1\%}\\[-1pt] \scriptsize p90: 0.97\%\\[-1pt] \scriptsize p95: 1.29\%} \\
E-n101-k8 & \makecell[tc]{\textbf{57.8\%}\\[-1pt] \scriptsize p90: 0.40\%\\[-1pt] \scriptsize p95: 0.51\%} & \makecell[tc]{\textbf{43.5\%}\\[-1pt] \scriptsize p90: 0.42\%\\[-1pt] \scriptsize p95: 0.59\%} & \makecell[tc]{\textbf{25.4\%}\\[-1pt] \scriptsize p90: 0.74\%\\[-1pt] \scriptsize p95: 0.97\%} & \makecell[tc]{\textbf{15.8\%}\\[-1pt] \scriptsize p90: 0.81\%\\[-1pt] \scriptsize p95: 1.01\%} & \makecell[tc]{\textbf{19.1\%}\\[-1pt] \scriptsize p90: 0.82\%\\[-1pt] \scriptsize p95: 1.01\%} \\
X-n143-k7 & \makecell[tc]{\textbf{56.3\%}\\[-1pt] \scriptsize p90: 0.39\%\\[-1pt] \scriptsize p95: 0.55\%} & \makecell[tc]{\textbf{54.1\%}\\[-1pt] \scriptsize p90: 0.39\%\\[-1pt] \scriptsize p95: 0.56\%} & \makecell[tc]{\textbf{42.5\%}\\[-1pt] \scriptsize p90: 0.56\%\\[-1pt] \scriptsize p95: 0.76\%} & \makecell[tc]{\textbf{13.5\%}\\[-1pt] \scriptsize p90: 0.96\%\\[-1pt] \scriptsize p95: 1.19\%} & \makecell[tc]{\textbf{16.3\%}\\[-1pt] \scriptsize p90: 0.93\%\\[-1pt] \scriptsize p95: 1.17\%} \\
X-n214-k11 & \makecell[tc]{\textbf{46.6\%}\\[-1pt] \scriptsize p90: 0.27\%\\[-1pt] \scriptsize p95: 0.38\%} & \makecell[tc]{\textbf{40.4\%}\\[-1pt] \scriptsize p90: 0.28\%\\[-1pt] \scriptsize p95: 0.38\%} & \makecell[tc]{\textbf{16.8\%}\\[-1pt] \scriptsize p90: 0.56\%\\[-1pt] \scriptsize p95: 0.73\%} & \makecell[tc]{\textbf{12.2\%}\\[-1pt] \scriptsize p90: 0.53\%\\[-1pt] \scriptsize p95: 0.67\%} & \makecell[tc]{\textbf{13.3\%}\\[-1pt] \scriptsize p90: 0.56\%\\[-1pt] \scriptsize p95: 0.73\%} \\
X-n351-k40 & \makecell[tc]{\textbf{35.1\%}\\[-1pt] \scriptsize p90: 0.11\%\\[-1pt] \scriptsize p95: 0.15\%} & \makecell[tc]{\textbf{13.5\%}\\[-1pt] \scriptsize p90: 0.19\%\\[-1pt] \scriptsize p95: 0.26\%} & \makecell[tc]{\textbf{9.1\%}\\[-1pt] \scriptsize p90: 0.25\%\\[-1pt] \scriptsize p95: 0.32\%} & \makecell[tc]{\textbf{1.6\%}\\[-1pt] \scriptsize p90: 0.36\%\\[-1pt] \scriptsize p95: 0.43\%} & \makecell[tc]{\textbf{2.3\%}\\[-1pt] \scriptsize p90: 0.36\%\\[-1pt] \scriptsize p95: 0.43\%} \\
X-n459-k26 & \makecell[tc]{\textbf{30.2\%}\\[-1pt] \scriptsize p90: 0.17\%\\[-1pt] \scriptsize p95: 0.22\%} & \makecell[tc]{\textbf{19.8\%}\\[-1pt] \scriptsize p90: 0.18\%\\[-1pt] \scriptsize p95: 0.24\%} & \makecell[tc]{\textbf{8.5\%}\\[-1pt] \scriptsize p90: 0.33\%\\[-1pt] \scriptsize p95: 0.43\%} & \makecell[tc]{\textbf{2.6\%}\\[-1pt] \scriptsize p90: 0.37\%\\[-1pt] \scriptsize p95: 0.45\%} & \makecell[tc]{\textbf{2.9\%}\\[-1pt] \scriptsize p90: 0.40\%\\[-1pt] \scriptsize p95: 0.50\%} \\
X-n573-k30 & \makecell[tc]{\textbf{18.5\%}\\[-1pt] \scriptsize p90: 0.11\%\\[-1pt] \scriptsize p95: 0.15\%} & \makecell[tc]{\textbf{18.2\%}\\[-1pt] \scriptsize p90: 0.11\%\\[-1pt] \scriptsize p95: 0.15\%} & \makecell[tc]{\textbf{14.8\%}\\[-1pt] \scriptsize p90: 0.13\%\\[-1pt] \scriptsize p95: 0.17\%} & \makecell[tc]{\textbf{4.0\%}\\[-1pt] \scriptsize p90: 0.34\%\\[-1pt] \scriptsize p95: 0.41\%} & \makecell[tc]{\textbf{4.4\%}\\[-1pt] \scriptsize p90: 0.34\%\\[-1pt] \scriptsize p95: 0.42\%} \\
X-n685-k75 & \makecell[tc]{\textbf{14.9\%}\\[-1pt] \scriptsize p90: 0.07\%\\[-1pt] \scriptsize p95: 0.09\%} & \makecell[tc]{\textbf{7.6\%}\\[-1pt] \scriptsize p90: 0.09\%\\[-1pt] \scriptsize p95: 0.11\%} & \makecell[tc]{\textbf{3.3\%}\\[-1pt] \scriptsize p90: 0.10\%\\[-1pt] \scriptsize p95: 0.13\%} & \makecell[tc]{\textbf{0.0\%}\\[-1pt] \scriptsize p90: 2.63\%\\[-1pt] \scriptsize p95: 2.71\%} & \makecell[tc]{\textbf{0.0\%}\\[-1pt] \scriptsize p90: 2.61\%\\[-1pt] \scriptsize p95: 2.70\%} \\
X-n749-k98 & \makecell[tc]{\textbf{10.4\%}\\[-1pt] \scriptsize p90: 0.05\%\\[-1pt] \scriptsize p95: 0.06\%} & \makecell[tc]{\textbf{5.4\%}\\[-1pt] \scriptsize p90: 0.06\%\\[-1pt] \scriptsize p95: 0.09\%} & \makecell[tc]{\textbf{3.3\%}\\[-1pt] \scriptsize p90: 0.08\%\\[-1pt] \scriptsize p95: 0.10\%} & \makecell[tc]{\textbf{0.0\%}\\[-1pt] \scriptsize p90: 0.35\%\\[-1pt] \scriptsize p95: 0.38\%} & \makecell[tc]{\textbf{0.0\%}\\[-1pt] \scriptsize p90: 0.35\%\\[-1pt] \scriptsize p95: 0.38\%} \\
X-n819-k171 & \makecell[tc]{\textbf{46.6\%}\\[-1pt] \scriptsize p90: 0.01\%\\[-1pt] \scriptsize p95: 0.01\%} & \makecell[tc]{\textbf{27.9\%}\\[-1pt] \scriptsize p90: 0.02\%\\[-1pt] \scriptsize p95: 0.02\%} & \makecell[tc]{\textbf{19.8\%}\\[-1pt] \scriptsize p90: 0.02\%\\[-1pt] \scriptsize p95: 0.03\%} & \makecell[tc]{\textbf{0.0\%}\\[-1pt] \scriptsize p90: 0.72\%\\[-1pt] \scriptsize p95: 0.77\%} & \makecell[tc]{\textbf{0.0\%}\\[-1pt] \scriptsize p90: 0.72\%\\[-1pt] \scriptsize p95: 0.77\%} \\
X-n916-k207 & \makecell[tc]{\textbf{27.1\%}\\[-1pt] \scriptsize p90: 0.01\%\\[-1pt] \scriptsize p95: 0.01\%} & \makecell[tc]{\textbf{25.1\%}\\[-1pt] \scriptsize p90: 0.01\%\\[-1pt] \scriptsize p95: 0.01\%} & \makecell[tc]{\textbf{20.2\%}\\[-1pt] \scriptsize p90: 0.01\%\\[-1pt] \scriptsize p95: 0.02\%} & \makecell[tc]{\textbf{0.0\%}\\[-1pt] \scriptsize p90: 0.68\%\\[-1pt] \scriptsize p95: 0.70\%} & \makecell[tc]{\textbf{0.0\%}\\[-1pt] \scriptsize p90: 0.68\%\\[-1pt] \scriptsize p95: 0.70\%} \\
X-n1001-k43 & \makecell[tc]{\textbf{12.5\%}\\[-1pt] \scriptsize p90: 0.12\%\\[-1pt] \scriptsize p95: 0.16\%} & \makecell[tc]{\textbf{11.3\%}\\[-1pt] \scriptsize p90: 0.13\%\\[-1pt] \scriptsize p95: 0.17\%} & \makecell[tc]{\textbf{5.7\%}\\[-1pt] \scriptsize p90: 0.18\%\\[-1pt] \scriptsize p95: 0.23\%} & \makecell[tc]{\textbf{0.0\%}\\[-1pt] \scriptsize p90: 0.96\%\\[-1pt] \scriptsize p95: 1.05\%} & \makecell[tc]{\textbf{0.0\%}\\[-1pt] \scriptsize p90: 0.97\%\\[-1pt] \scriptsize p95: 1.04\%} \\
\bottomrule
\end{tabular}
\end{table}

To compare these distributions fairly across methods and instances, Figure~\ref{fig:ecdf} shows the \textit{empirical cumulative distribution function} (\textbf{ECDF}) of the performance gaps for every method on every instance. For a given value on the x-axis, the ECDF gives the proportion of solved permutations whose gap is at most that value. Hence, curves that rise faster and lie further to the left indicate better performance, since they place more probability mass on smaller gaps. The x-axis is shown on a logarithmic scale to accommodate the wide range of observed gap magnitudes. The ECDF of the optimal FP-FLA method is shown as a red reference line; the closer another method is to this line, the better its performance.

Overall, after FP-FLA, which is optimal by construction, FR-FLA/FRVCPY performs best, followed by SS-FR-FLA. BHGA is generally competitive and often not far behind, whereas BACO and CACO are only slightly worse on smaller instances but deteriorate substantially on larger ones. To put these differences into perspective, BHGA improved the best known solution for X-n916-k207 by 1.77\%, which is considered a substantial improvement. Thus, even seemingly small percentage differences can be practically important, especially on a benchmark set such as WCCI-2020, where solutions have already been heavily optimized.

\begin{figure*}[htbp]
    \centering
    \includegraphics[width=\textwidth]{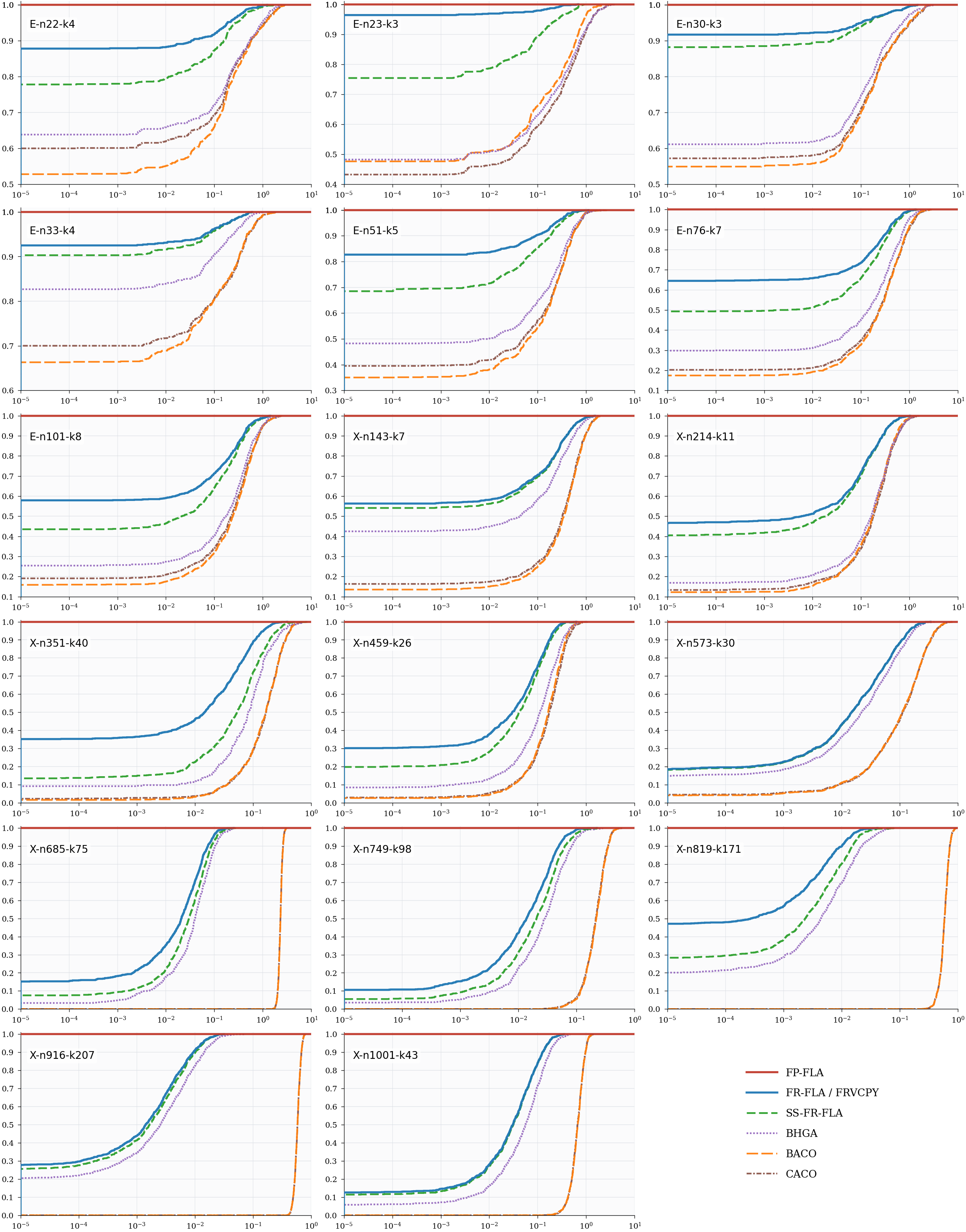}
    \caption{Empirical cumulative distribution functions of percentage performance gaps for all methods and WCCI-2020 benchmark instances using stochastic kNN permutations. Gaps are computed relative to the optimal FP-FLA decoding and only over permutations successfully decoded by each method. For a given x-axis value, the curve gives the proportion of decoded permutations with a performance gap no larger than that value. Curves that rise earlier and lie further left indicate better solution quality. The x-axis is logarithmic to show both small and large gaps.}
    \label{fig:ecdf}
\end{figure*}

\subsubsection{Decoding Time Analysis}

The final performance aspect to consider is the time required to decode a permutation. For methods not developed in this work, we measured the runtime of the open-source implementations provided by their authors. The results are shown in Figure~\ref{fig:time}.

Among all methods, FRVCPY is by far the slowest. This is expected, since it solves a more general problem that allows partial recharges. In comparison, our equivalent FR-FLA is substantially faster. FP-FLA requires more time than FR-FLA, whereas SS-FR-FLA is faster, which is also consistent with the relative complexity of the corresponding decoding procedures. The three heuristic methods are the fastest overall and have roughly comparable execution times.

The figure also reveals that the runtime gap between our methods and the heuristic decoders narrows as instance size increases. On medium and large instances, SS-FR-FLA performs roughly on par with the heuristic methods in terms of median decoding time, which is particularly notable given that it retains a much more structured decoding procedure.

Although the heuristic methods generally achieve lower median decoding times, SS-FR-FLA's and FR-FLA's third quartiles are often quite similar to those of the heuristics. This suggests that while their typical-case performance differs, their slower decoding cases are of comparable magnitude. In other words, the differences between these methods are more visible in the median behavior than in the upper part of the runtime distribution.

The sensitivity of our three methods to the number and spatial distribution of charging stations is clearly visible in instances such as X-n351-k40, X-n685-k75, and X-n749-k98, which contain a high density of charging stations. On these instances, the differences in runtime between FP-FLA, FR-FLA, and SS-FR-FLA become particularly pronounced. This indicates that decoding time is influenced not only by the number of customers, but also by structural properties of the instance itself.

These runtime differences are practically important because a decoder is typically invoked a very large number of times within a metaheuristic search. Consequently, even modest differences in per-permutation runtime can accumulate into substantial differences in total optimization time.

\begin{figure*}[htbp]
    \centering
    \includegraphics[width=\textwidth]{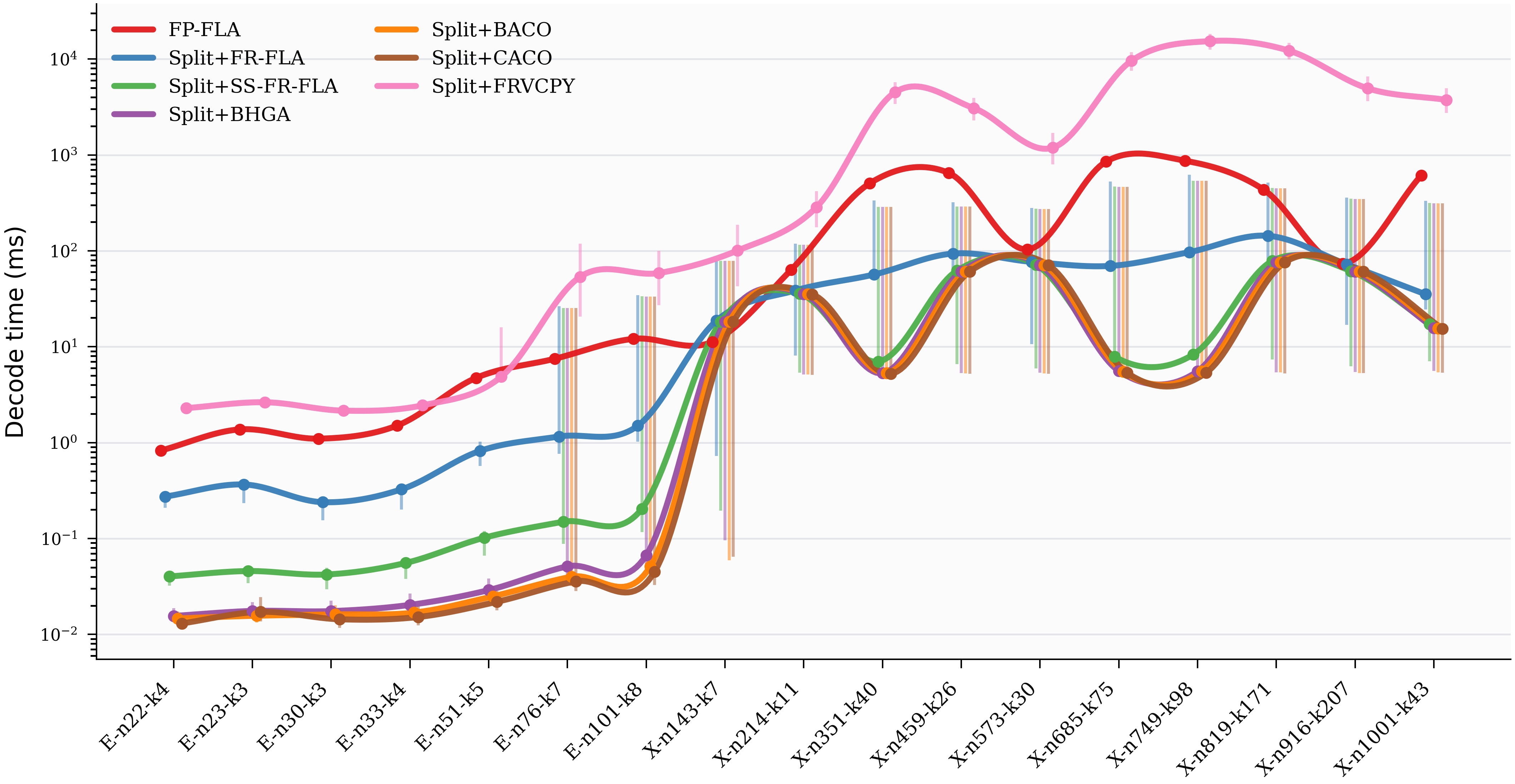}
    \caption{Median permutation decoding time on the WCCI-2020 benchmark instances using stochastic kNN \(k=2\) permutations. Instances are ordered by the number of customers. Points show the median decoding time, and vertical bars span the first to third quartiles of the measured decoding times. The y-axis is logarithmic.}
    \label{fig:time}
\end{figure*}

\subsection{Discussion and Practical Implications}

Taken together, the results show a clear trade-off between decoding time and solution quality, but also show that this trade-off is not equally favorable for all methods. Among all tested methods, SS-FR-FLA seems to offer the best balance for use inside a heuristic optimization loop. Its runtime is close to that of the heuristic decoders, especially on larger instances, while its performance gap is much smaller and the percentage of solved instances remains high. This makes it the most promising candidate for repeated use inside a metaheuristic. FR-FLA is also relevant, especially in the cases where SS-FR-FLA cannot decode a permutation or when a smaller performance gap is desired. However, the gain in solution quality over SS-FR-FLA is usually not large enough to justify the runtime increase. Because of that, FR-FLA does not seem as attractive as the default decoder inside the optimization loop, but it can still be useful in selected situations. FP-FLA gives the strongest guarantee since it always finds the optimal decoding for a given permutation, but it is probably too slow to be used as the main decoder inside the optimization loop. Still, with decode times below one second on the tested instances, it remains practical for improving final solutions, validating the quality of other decoders, or solving particularly difficult cases where exact decoding is justified.

Another important observation is that instances and permutations for which split decoding is not optimal seem to be much more common than is usually suggested. The measured performance gaps are often small in relative terms, but they are systematic and clearly visible across many instances. The results also show a clear limitation of the heuristic solvers. While their decoding times are low, they fail to solve a substantial fraction of permutations on many instances, especially in the random ordering setting and on the larger instances. This is an important drawback in the context of heuristic optimization, because an unsolved permutation eliminates a whole subset of solutions from ever being considered. In that sense, the percentage of solved instances is not just a secondary metric, but an important part of decoder robustness.

Overall, the experiments suggest a clear practical picture. SS-FR-FLA appears to be the best default choice for heuristic optimization, since it combines near-heuristic runtime with substantially better robustness and solution quality. FR-FLA can still be useful when some additional improvement is needed, while FP-FLA is best suited for final improvement, validation, and hard cases.

%% file: conclusion.tex
\section{Conclusion}
\label{sec:conclusion}

Permutation-based metaheuristics for the EVRP require a decoding step that transforms a customer ordering into a complete, feasible solution. This paper addressed the decoding problem directly by introducing the Fixed-Permutation Splitting and Charging Problem, which captures the joint nature of route splitting and charging-station insertion from a fixed customer ordering. We proposed FP-FLA, an exact forward labeling algorithm that solves the FPSCP optimally within a dynamic programming framework with dominance pruning, and proved its optimality. We further derived two restricted variants, a split-then-charge decoder and a single-station variant, that correspond to simplifications commonly found in the literature, and compared all proposed methods against three heuristic FRVCP decoders on the WCCI-2020 benchmark.

The experiments revealed that common decoding simplifications carry a measurable cost. Decoupling route splitting from charging-station insertion leads to systematic, if often small, losses in solution quality relative to joint decoding. Restricting the number of charging-station visits to at most one per inter-customer segment reduces robustness further and, on some instances, excludes feasible solutions entirely. The heuristic decoders are the fastest, but fail to decode a substantial fraction of feasible permutations, which is a significant drawback when every unsolved permutation effectively removes a candidate solution from the search space.

Despite its larger theoretical branching factor, FP-FLA remained tractable on all tested instances, with Pareto fronts staying well-bounded throughout decoding. Among the proposed variants, SS-FR-FLA appears best suited for repeated use inside a metaheuristic: it matches heuristic decoder runtimes on medium and large instances while offering substantially better robustness and solution quality. FR-FLA provides a useful intermediate option when higher decoding accuracy is needed, and FP-FLA is well suited for solution validation, final improvement, and cases where exact decoding is justified.

Future work may study how the decoder-level trade-offs characterized here translate into full metaheuristic performance. On the algorithmic side, promising directions include extending the framework to richer EVRP variants such as partial recharging, time windows, or nonlinear charging functions, as well as developing additional pruning rules and preprocessing techniques to further reduce decoding time while preserving exactness.

%% file: appendix.tex
\clearpage
\appendix

\section{Results for uniformly random permutations}
\label{app:random-permutations}

This appendix reports additional experiments using uniformly random customer permutations. In contrast to the stochastic kNN permutations
used in the main text, uniformly random permutations are less representative of solutions typically encountered in a search procedure. However, they provide a useful stress test for decoder robustness, especially in early search stages when candidate permutations may be poorly structured.

The experimental setup is the same as in Section~\ref{sec:results}. For each WCCI-2020 benchmark instance, 1000 uniformly random customer permutations were generated. All performance gaps are computed relative to the optimal FP-FLA decoding of the same permutation and only over permutations successfully decoded by the method being evaluated. To keep the presentation compact and coherent, all appendix results are presented first, followed by a short discussion at the end of the appendix.

\begin{table}[htbp]
\centering
\small
\setlength{\tabcolsep}{5pt}
\caption{Percentage of feasible uniformly random permutations successfully decoded by each method on the WCCI-2020 benchmark instances. Percentages are computed only over permutations for which a feasible solution exists.}
\label{tab:solved_random}
\begin{tabular}{lccccc}
\toprule
Instance & \makecell{FR-FLA /\\ FRVCPY} & SS-FR-FLA & BHGA & BACO & CACO \\
\midrule
E-n22-k4     & 100.0\% & 100.0\% & 64.9\% & 94.5\%  & 0.0\% \\
E-n23-k3     & 100.0\% & 100.0\% & 56.9\% & 100.0\% & 3.1\% \\
E-n30-k3     & 100.0\% & 100.0\% & 43.3\% & 100.0\% & 0.4\% \\
E-n33-k4     & 100.0\% & 100.0\% & 99.3\% & 100.0\% & 0.0\% \\
E-n51-k5     & 100.0\% & 100.0\% & 42.6\% & 100.0\% & 0.0\% \\
E-n76-k7     & 100.0\% & 100.0\% & 2.7\%  & 93.3\%  & 0.0\% \\
E-n101-k8    & 100.0\% & 100.0\% & 1.2\%  & 96.1\%  & 0.0\% \\
X-n143-k7    & 100.0\% & 100.0\% & 40.2\% & 100.0\% & 0.0\% \\
X-n214-k11   & 100.0\% & 100.0\% & 36.7\% & 100.0\% & 0.0\% \\
X-n351-k40   & 100.0\% & 90.2\%  & 0.0\%  & 3.9\%   & 0.0\% \\
X-n459-k26   & 100.0\% & 100.0\% & 0.0\%  & 78.5\%  & 0.0\% \\
X-n573-k30   & 100.0\% & 100.0\% & 94.9\% & 100.0\% & 0.0\% \\
X-n685-k75   & 100.0\% & 97.8\%  & 0.0\%  & 0.9\%   & 0.0\% \\
X-n749-k98   & 100.0\% & 100.0\% & 0.0\%  & 24.5\%  & 0.0\% \\
X-n819-k171  & 100.0\% & 98.0\%  & 0.0\%  & 0.5\%   & 0.0\% \\
X-n916-k207  & 100.0\% & 100.0\% & 38.8\% & 100.0\% & 0.0\% \\
X-n1001-k43  & 100.0\% & 100.0\% & 0.0\%  & 100.0\% & 0.0\% \\
\midrule
\textbf{Aggregate}    & \textbf{100.0\%} & \textbf{99.2\%}  & \textbf{30.7\%} & \textbf{76.0\%}  & \textbf{0.2\%} \\
\bottomrule
\end{tabular}
\end{table}

\begin{table}[htbp]
\centering
\footnotesize
\setlength{\tabcolsep}{4pt}
\renewcommand{\arraystretch}{1.15}
\caption{Performance-gap summary for uniformly random permutations on the WCCI-2020 benchmark instances. Bold values denote the zero-gap percentage, defined as the share of successfully decoded permutations for which the method matches the optimal FP-FLA solution. The p90 and p95 values are the 90th and 95th percentiles of the percentage performance-gap distribution relative to FP-FLA, computed over successfully decoded permutations. Lower p90 and p95 values indicate smaller upper-tail losses.}
\label{tab:zero_p90_p95_random}
\begin{tabular}{lccccc}
\toprule
Instance & \makecell{FR-FLA/\\FRVCPY} & SS-FR-FLA & BHGA & BACO & CACO \\
\midrule
E-n22-k4 & \makecell[tc]{\textbf{83.4\%}\\[-1pt] \scriptsize p90: 0.02\%\\[-1pt] \scriptsize p95: 0.07\%} & \makecell[tc]{\textbf{50.8\%}\\[-1pt] \scriptsize p90: 0.17\%\\[-1pt] \scriptsize p95: 0.25\%} & \makecell[tc]{\textbf{17.9\%}\\[-1pt] \scriptsize p90: 1.05\%\\[-1pt] \scriptsize p95: 1.42\%} & \makecell[tc]{\textbf{12.0\%}\\[-1pt] \scriptsize p90: 1.07\%\\[-1pt] \scriptsize p95: 1.45\%} & \makecell[tc]{---} \\
E-n23-k3 & \makecell[tc]{\textbf{89.7\%}\\[-1pt] \scriptsize p90: 0.00\%\\[-1pt] \scriptsize p95: 0.03\%} & \makecell[tc]{\textbf{65.0\%}\\[-1pt] \scriptsize p90: 0.10\%\\[-1pt] \scriptsize p95: 0.16\%} & \makecell[tc]{\textbf{17.0\%}\\[-1pt] \scriptsize p90: 0.74\%\\[-1pt] \scriptsize p95: 0.98\%} & \makecell[tc]{\textbf{17.6\%}\\[-1pt] \scriptsize p90: 0.56\%\\[-1pt] \scriptsize p95: 0.78\%} & \makecell[tc]{\textbf{0.0\%}\\[-1pt] \scriptsize p90: 0.67\%\\[-1pt] \scriptsize p95: 0.77\%} \\
E-n30-k3 & \makecell[tc]{\textbf{61.9\%}\\[-1pt] \scriptsize p90: 0.03\%\\[-1pt] \scriptsize p95: 0.06\%} & \makecell[tc]{\textbf{42.9\%}\\[-1pt] \scriptsize p90: 0.12\%\\[-1pt] \scriptsize p95: 0.18\%} & \makecell[tc]{\textbf{0.2\%}\\[-1pt] \scriptsize p90: 1.36\%\\[-1pt] \scriptsize p95: 1.76\%} & \makecell[tc]{\textbf{0.1\%}\\[-1pt] \scriptsize p90: 1.40\%\\[-1pt] \scriptsize p95: 1.77\%} & \makecell[tc]{\textbf{0.0\%}\\[-1pt] \scriptsize p90: 1.68\%\\[-1pt] \scriptsize p95: 1.81\%} \\
E-n33-k4 & \makecell[tc]{\textbf{91.3\%}\\[-1pt] \scriptsize p90: 0.00\%\\[-1pt] \scriptsize p95: 0.08\%} & \makecell[tc]{\textbf{87.0\%}\\[-1pt] \scriptsize p90: 0.01\%\\[-1pt] \scriptsize p95: 0.09\%} & \makecell[tc]{\textbf{59.2\%}\\[-1pt] \scriptsize p90: 0.21\%\\[-1pt] \scriptsize p95: 0.33\%} & \makecell[tc]{\textbf{47.8\%}\\[-1pt] \scriptsize p90: 0.30\%\\[-1pt] \scriptsize p95: 0.37\%} & \makecell[tc]{---} \\
E-n51-k5 & \makecell[tc]{\textbf{47.9\%}\\[-1pt] \scriptsize p90: 0.07\%\\[-1pt] \scriptsize p95: 0.11\%} & \makecell[tc]{\textbf{25.1\%}\\[-1pt] \scriptsize p90: 0.12\%\\[-1pt] \scriptsize p95: 0.17\%} & \makecell[tc]{\textbf{2.1\%}\\[-1pt] \scriptsize p90: 0.73\%\\[-1pt] \scriptsize p95: 0.89\%} & \makecell[tc]{\textbf{1.5\%}\\[-1pt] \scriptsize p90: 0.61\%\\[-1pt] \scriptsize p95: 0.71\%} & \makecell[tc]{---} \\
E-n76-k7 & \makecell[tc]{\textbf{12.5\%}\\[-1pt] \scriptsize p90: 0.10\%\\[-1pt] \scriptsize p95: 0.14\%} & \makecell[tc]{\textbf{3.2\%}\\[-1pt] \scriptsize p90: 0.19\%\\[-1pt] \scriptsize p95: 0.23\%} & \makecell[tc]{\textbf{0.0\%}\\[-1pt] \scriptsize p90: 1.10\%\\[-1pt] \scriptsize p95: 1.13\%} & \makecell[tc]{\textbf{0.0\%}\\[-1pt] \scriptsize p90: 1.08\%\\[-1pt] \scriptsize p95: 1.29\%} & \makecell[tc]{---} \\
E-n101-k8 & \makecell[tc]{\textbf{13.1\%}\\[-1pt] \scriptsize p90: 0.08\%\\[-1pt] \scriptsize p95: 0.11\%} & \makecell[tc]{\textbf{3.2\%}\\[-1pt] \scriptsize p90: 0.13\%\\[-1pt] \scriptsize p95: 0.17\%} & \makecell[tc]{\textbf{0.0\%}\\[-1pt] \scriptsize p90: 1.16\%\\[-1pt] \scriptsize p95: 1.19\%} & \makecell[tc]{\textbf{0.0\%}\\[-1pt] \scriptsize p90: 0.94\%\\[-1pt] \scriptsize p95: 1.07\%} & \makecell[tc]{---} \\
X-n143-k7 & \makecell[tc]{\textbf{65.2\%}\\[-1pt] \scriptsize p90: 0.05\%\\[-1pt] \scriptsize p95: 0.07\%} & \makecell[tc]{\textbf{44.9\%}\\[-1pt] \scriptsize p90: 0.06\%\\[-1pt] \scriptsize p95: 0.08\%} & \makecell[tc]{\textbf{0.2\%}\\[-1pt] \scriptsize p90: 0.82\%\\[-1pt] \scriptsize p95: 0.94\%} & \makecell[tc]{\textbf{2.4\%}\\[-1pt] \scriptsize p90: 0.23\%\\[-1pt] \scriptsize p95: 0.28\%} & \makecell[tc]{---} \\
X-n214-k11 & \makecell[tc]{\textbf{35.4\%}\\[-1pt] \scriptsize p90: 0.07\%\\[-1pt] \scriptsize p95: 0.09\%} & \makecell[tc]{\textbf{23.1\%}\\[-1pt] \scriptsize p90: 0.07\%\\[-1pt] \scriptsize p95: 0.10\%} & \makecell[tc]{\textbf{0.0\%}\\[-1pt] \scriptsize p90: 0.74\%\\[-1pt] \scriptsize p95: 0.84\%} & \makecell[tc]{\textbf{0.3\%}\\[-1pt] \scriptsize p90: 0.30\%\\[-1pt] \scriptsize p95: 0.36\%} & \makecell[tc]{---} \\
X-n351-k40 & \makecell[tc]{\textbf{2.4\%}\\[-1pt] \scriptsize p90: 0.03\%\\[-1pt] \scriptsize p95: 0.04\%} & \makecell[tc]{\textbf{0.0\%}\\[-1pt] \scriptsize p90: 0.25\%\\[-1pt] \scriptsize p95: 0.27\%} & \makecell[tc]{---} & \makecell[tc]{\textbf{0.0\%}\\[-1pt] \scriptsize p90: 0.79\%\\[-1pt] \scriptsize p95: 0.81\%} & \makecell[tc]{---} \\
X-n459-k26 & \makecell[tc]{\textbf{0.4\%}\\[-1pt] \scriptsize p90: 0.02\%\\[-1pt] \scriptsize p95: 0.03\%} & \makecell[tc]{\textbf{0.0\%}\\[-1pt] \scriptsize p90: 0.07\%\\[-1pt] \scriptsize p95: 0.08\%} & \makecell[tc]{---} & \makecell[tc]{\textbf{0.0\%}\\[-1pt] \scriptsize p90: 0.49\%\\[-1pt] \scriptsize p95: 0.53\%} & \makecell[tc]{---} \\
X-n573-k30 & \makecell[tc]{\textbf{2.7\%}\\[-1pt] \scriptsize p90: 0.15\%\\[-1pt] \scriptsize p95: 0.18\%} & \makecell[tc]{\textbf{2.5\%}\\[-1pt] \scriptsize p90: 0.15\%\\[-1pt] \scriptsize p95: 0.18\%} & \makecell[tc]{\textbf{0.0\%}\\[-1pt] \scriptsize p90: 0.43\%\\[-1pt] \scriptsize p95: 0.50\%} & \makecell[tc]{\textbf{0.0\%}\\[-1pt] \scriptsize p90: 0.37\%\\[-1pt] \scriptsize p95: 0.42\%} & \makecell[tc]{---} \\
X-n685-k75 & \makecell[tc]{\textbf{0.0\%}\\[-1pt] \scriptsize p90: 0.02\%\\[-1pt] \scriptsize p95: 0.02\%} & \makecell[tc]{\textbf{0.0\%}\\[-1pt] \scriptsize p90: 0.15\%\\[-1pt] \scriptsize p95: 0.16\%} & \makecell[tc]{---} & \makecell[tc]{\textbf{0.0\%}\\[-1pt] \scriptsize p90: 0.69\%\\[-1pt] \scriptsize p95: 0.70\%} & \makecell[tc]{---} \\
X-n749-k98 & \makecell[tc]{\textbf{1.0\%}\\[-1pt] \scriptsize p90: 0.02\%\\[-1pt] \scriptsize p95: 0.02\%} & \makecell[tc]{\textbf{0.0\%}\\[-1pt] \scriptsize p90: 0.08\%\\[-1pt] \scriptsize p95: 0.09\%} & \makecell[tc]{---} & \makecell[tc]{\textbf{0.0\%}\\[-1pt] \scriptsize p90: 0.50\%\\[-1pt] \scriptsize p95: 0.54\%} & \makecell[tc]{---} \\
X-n819-k171 & \makecell[tc]{\textbf{0.1\%}\\[-1pt] \scriptsize p90: 0.01\%\\[-1pt] \scriptsize p95: 0.01\%} & \makecell[tc]{\textbf{0.0\%}\\[-1pt] \scriptsize p90: 0.09\%\\[-1pt] \scriptsize p95: 0.10\%} & \makecell[tc]{---} & \makecell[tc]{\textbf{0.0\%}\\[-1pt] \scriptsize p90: 0.48\%\\[-1pt] \scriptsize p95: 0.51\%} & \makecell[tc]{---} \\
X-n916-k207 & \makecell[tc]{\textbf{4.6\%}\\[-1pt] \scriptsize p90: 0.01\%\\[-1pt] \scriptsize p95: 0.02\%} & \makecell[tc]{\textbf{0.0\%}\\[-1pt] \scriptsize p90: 0.02\%\\[-1pt] \scriptsize p95: 0.03\%} & \makecell[tc]{\textbf{0.0\%}\\[-1pt] \scriptsize p90: 0.11\%\\[-1pt] \scriptsize p95: 0.12\%} & \makecell[tc]{\textbf{0.0\%}\\[-1pt] \scriptsize p90: 0.18\%\\[-1pt] \scriptsize p95: 0.20\%} & \makecell[tc]{---} \\
X-n1001-k43 & \makecell[tc]{\textbf{14.1\%}\\[-1pt] \scriptsize p90: 0.01\%\\[-1pt] \scriptsize p95: 0.01\%} & \makecell[tc]{\textbf{0.0\%}\\[-1pt] \scriptsize p90: 0.02\%\\[-1pt] \scriptsize p95: 0.02\%} & \makecell[tc]{---} & \makecell[tc]{\textbf{0.0\%}\\[-1pt] \scriptsize p90: 0.17\%\\[-1pt] \scriptsize p95: 0.18\%} & \makecell[tc]{---} \\
\bottomrule
\end{tabular}
\end{table}

\begin{figure*}[htbp]
    \centering
    \includegraphics[width=\textwidth]{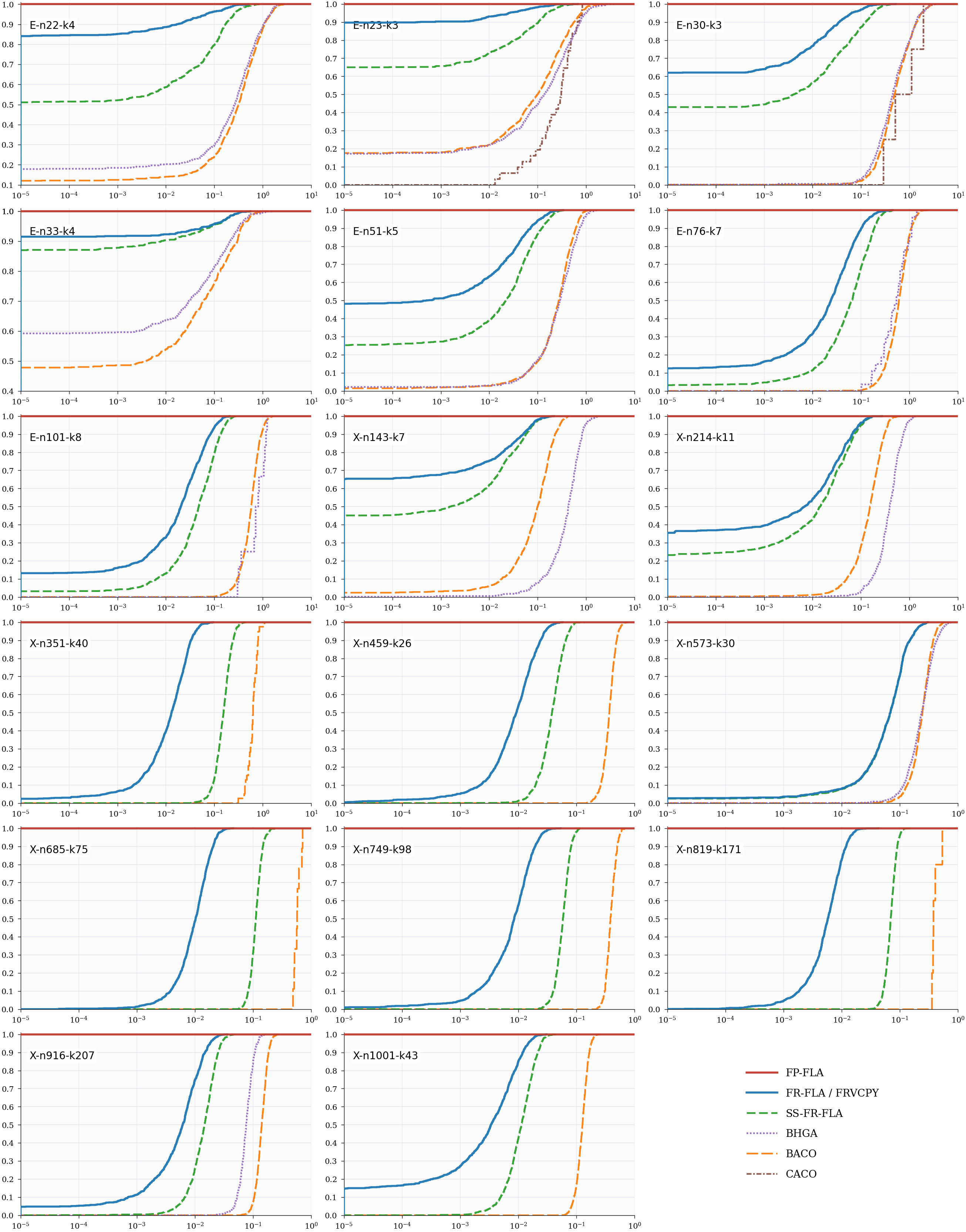}
    \caption{Empirical cumulative distribution functions of percentage performance gaps for all methods and WCCI-2020 benchmark instances using uniformly random permutations. Gaps are computed relative to the optimal FP-FLA decoding and only over permutations successfully decoded by each method. For a given x-axis value, the curve gives the proportion of decoded permutations with a performance gap no larger than that value. Curves that rise earlier and lie further left indicate better solution quality. The x-axis is logarithmic to show both small and large gaps.}
    \label{fig:rand-ecdf}
\end{figure*}

\begin{figure*}[htbp]
    \centering
    \includegraphics[width=\textwidth]{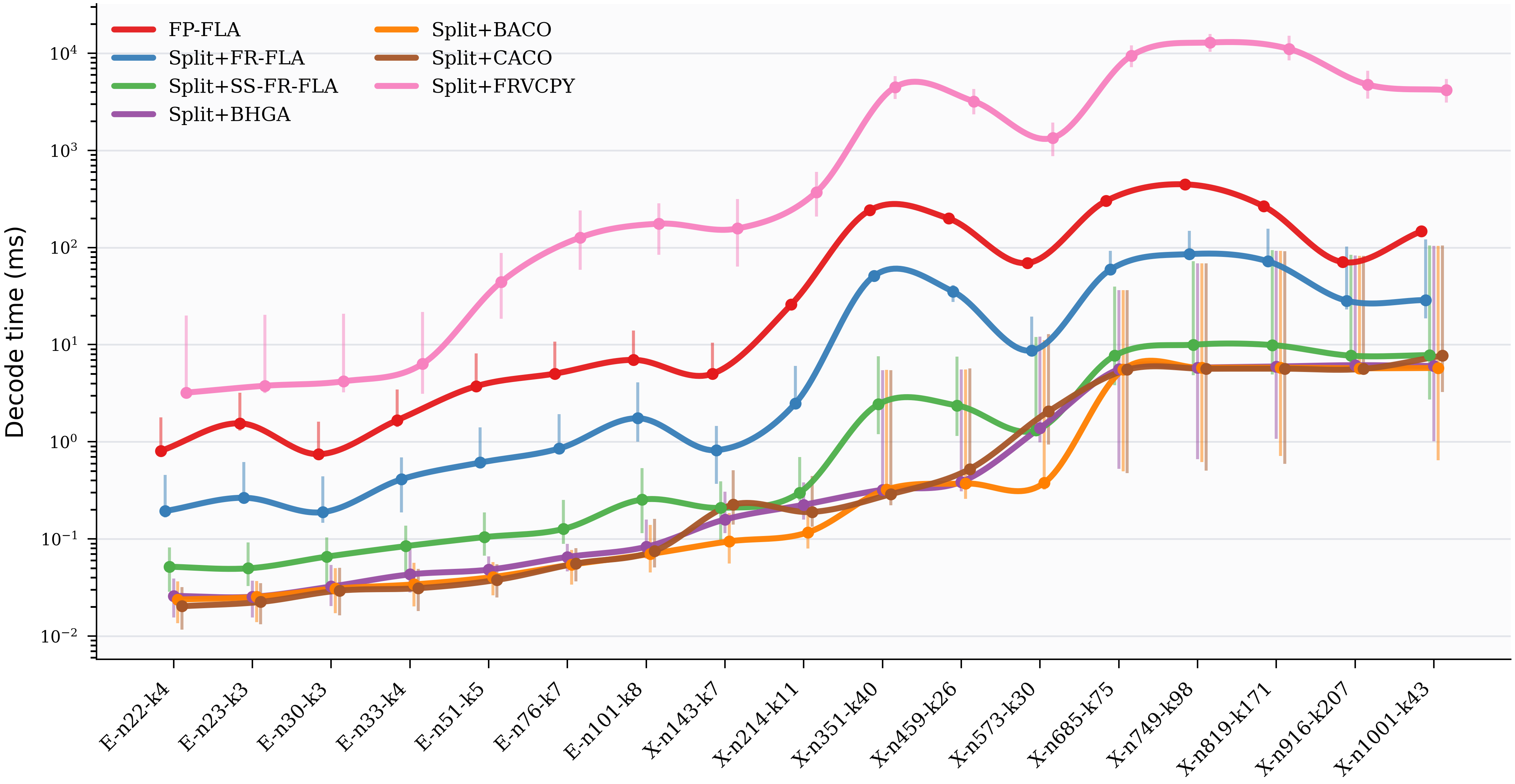}
    \caption{Median permutation decoding time on the WCCI-2020 benchmark instances using uniformly random permutations. Instances are ordered by the number of customers. Points show the median decoding time, and vertical bars span the first to third quartiles of the measured decoding times. The y-axis is logarithmic.}
    \label{fig:rand-time}
\end{figure*}

\FloatBarrier
The uniformly random permutation results should be interpreted primarily as a decoder stress test. Unlike the stochastic kNN permutations used in the main text, uniformly random permutations often contain long jumps between consecutive customers. This makes charging feasibility harder to restore and amplifies the effect of decoding simplifications. Consequently, the performance losses caused by simplifications are more pronounced than in the main experiments.

The solved-permutation results are particularly important in this setting. FP-FLA remains the reference method by construction, while FR-FLA successfully decodes all sampled permutations. The single-station variant also remains highly robust overall, although the restriction to at most one charging station between consecutive customers prevents it from solving some feasible permutations on the more difficult instances. In contrast, the heuristic decoders often have very low success rate, and in several cases fail to decode any of the feasible permutations. For such method--instance pairs, the corresponding performance-gap and runtime statistics should be interpreted with caution, since they are computed only over the small subset of permutations that the method successfully decoded.

The runtime results are affected by the same issue. When a heuristic decoder solves only a small fraction of permutations, its reported runtime reflects only those successfully solved cases, which are likely to be easier than the full set of feasible permutations. Therefore, the apparently larger runtime gap between the proposed labeling methods and the heuristic decoders should not be interpreted independently of the solved-permutation percentages. 

Overall, the uniformly random experiments reinforce the conclusions from the main text: stronger decoding methods are more robust and produce better solutions, while simpler and heuristic decoders reduce computational effort at the cost of feasibility and solution quality, especially when the input permutation is poorly structured.

%% file: references.bib
@article{jonidi2021urban,
  title={Urban air pollution control policies and strategies: a systematic review},
  author={Jonidi Jafari, Ahmad and Charkhloo, Esmail and Pasalari, Hasan},
  journal={Journal of Environmental Health Science and Engineering},
  volume={19},
  number={2},
  pages={1911--1940},
  year={2021},
  publisher={Springer}
}

@article{bhatti2020commerce,
  title={E-commerce trends during COVID-19 Pandemic},
  author={Bhatti, Anam and Akram, Hamza and Basit, Hafiz Muhammad and Khan, Ahmed Usman and Raza, Syeda Mahwish and Naqvi, Muhammad Bilal and others},
  journal={International Journal of Future Generation Communication and Networking},
  volume={13},
  number={2},
  pages={1449--1452},
  year={2020}
}

@article{juan2016electric,
  title={Electric vehicles in logistics and transportation: A survey on emerging environmental, strategic, and operational challenges},
  author={Juan, Angel Alejandro and Mendez, Carlos Alberto and Faulin, Javier and De Armas, Jesica and Grasman, Scott Erwin},
  journal={Energies},
  volume={9},
  number={2},
  pages={86},
  year={2016},
  publisher={MDPI}
}

@article{oliveira2017sustainable,
  title={Sustainable vehicles-based alternatives in last mile distribution of urban freight transport: A systematic literature review},
  author={Oliveira, Cintia Machado de and Albergaria De Mello Bandeira, Renata and Vasconcelos Goes, George and Schmitz Gon{\c{c}}alves, Daniel Neves and D’Agosto, M{\'a}rcio De Almeida},
  journal={Sustainability},
  volume={9},
  number={8},
  pages={1324},
  year={2017},
  publisher={MDPI}
}

@article{braekers2016vehicle,
  title={The vehicle routing problem: State of the art classification and review},
  author={Braekers, Kris and Ramaekers, Katrien and Van Nieuwenhuyse, Inneke},
  journal={Computers \& industrial engineering},
  volume={99},
  pages={300--313},
  year={2016},
  publisher={Elsevier}
}

@article{eksioglu2009vehicle,
  title={The vehicle routing problem: A taxonomic review},
  author={Eksioglu, Burak and Vural, Arif Volkan and Reisman, Arnold},
  journal={Computers \& Industrial Engineering},
  volume={57},
  number={4},
  pages={1472--1483},
  year={2009},
  publisher={Elsevier}
}

@article{kucukoglu2021electric,
  title={The electric vehicle routing problem and its variations: A literature review},
  author={Kucukoglu, Ilker and Dewil, Reginald and Cattrysse, Dirk},
  journal={Computers \& Industrial Engineering},
  volume={161},
  pages={107650},
  year={2021},
  publisher={Elsevier}
}

@article{erdelic2019survey,
  title={A survey on the electric vehicle routing problem: variants and solution approaches},
  author={Erdeli{\'c}, Tomislav and Cari{\'c}, Ton{\v{c}}i},
  journal={Journal of Advanced Transportation},
  volume={2019},
  number={1},
  pages={5075671},
  year={2019},
  publisher={Wiley Online Library}
}

@inproceedings{feng_bilevel_2024,
	title = {A Bilevel Hybrid Genetic Algorithm for Capacitated Electric Vehicle Routing Problem},
	url = {https://ieeexplore.ieee.org/abstract/document/10611987},
	doi = {10.1109/CEC60901.2024.10611987},
	eventtitle = {2024 {IEEE} Congress on Evolutionary Computation ({CEC})},
	pages = {1--8},
	booktitle = {2024 {IEEE} Congress on Evolutionary Computation ({CEC})},
	author = {Feng, Chang-Tao and Jia, Ya-Hui and Yang, Qiang and Chen, Wei-Neng and Jiang, Huaiguang},
	date = {2024-06},
    year={2024}
}

@article{jia2021bilevel,
  title={A bilevel ant colony optimization algorithm for capacitated electric vehicle routing problem},
  author={Jia, Ya-Hui and Mei, Yi and Zhang, Mengjie},
  journal={IEEE transactions on cybernetics},
  volume={52},
  number={10},
  pages={10855--10868},
  year={2021},
  publisher={IEEE}
}

@article{jia2022confidence,
  title={Confidence-based ant colony optimization for capacitated electric vehicle routing problem with comparison of different encoding schemes},
  author={Jia, Ya-Hui and Mei, Yi and Zhang, Mengjie},
  journal={IEEE Transactions on Evolutionary Computation},
  volume={26},
  number={6},
  pages={1394--1408},
  year={2022},
  publisher={IEEE}
}

@phdthesis{montoya2016electric,
  title={Electric Vehicle Routing Problems: models and solution approaches},
  author={Montoya, Alejandro},
  year={2016},
  school={Universit{\'e} d'Angers}
}

@article{pavai2016survey,
  title={A survey on crossover operators},
  author={Pavai, G and Geetha, TV},
  journal={ACM Computing Surveys (CSUR)},
  volume={49},
  number={4},
  pages={1--43},
  year={2016},
  publisher={ACM New York, NY, USA}
}

@article{larranaga1999genetic,
  title={Genetic algorithms for the travelling salesman problem: A review of representations and operators},
  author={Larranaga, Pedro and Kuijpers, Cindy M. H. and Murga, Roberto H. and Inza, Inaki and Dizdarevic, Sejla},
  journal={Artificial intelligence review},
  volume={13},
  number={2},
  pages={129--170},
  year={1999},
  publisher={Springer}
}

@inproceedings{mavrovouniotis2020benchmark,
  title={A benchmark test suite for the electric capacitated vehicle routing problem},
  author={Mavrovouniotis, Michalis and Menelaou, Charalambos and Timotheou, Stelios and Ellinas, Georgios and Panayiotou, Christos and Polycarpou, Marios},
  booktitle={2020 IEEE Congress on evolutionary computation (CEC)},
  pages={1--8},
  year={2020},
  organization={IEEE}
}

@article{vidal2022hybrid,
  title={Hybrid genetic search for the CVRP: Open-source implementation and SWAP* neighborhood},
  author={Vidal, Thibaut},
  journal={Computers \& Operations Research},
  volume={140},
  pages={105643},
  year={2022},
  publisher={Elsevier}
}

@INPROCEEDINGS{9504893,
  author={Chitty, Darren M.},
  booktitle={2021 IEEE Congress on Evolutionary Computation (CEC)}, 
  title={An Ant Colony Optimisation Inspired Crossover Operator for Permutation Type Problems}, 
  year={2021},
  volume={},
  number={},
  pages={57-64},
  doi={10.1109/CEC45853.2021.9504893}
}

@article{tian1999application,
  title={Application of the simulated annealing algorithm to the combinatorial optimisation problem with permutation property: An investigation of generation mechanism},
  author={Tian, Peng and Ma, Jian and Zhang, Dong-Mo},
  journal={European Journal of Operational Research},
  volume={118},
  number={1},
  pages={81--94},
  year={1999},
  publisher={Elsevier}
}

@inproceedings{deschenes2020fixed,
  title={The fixed route electric vehicle charging problem with nonlinear energy management and variable vehicle speed},
  author={Desch{\^e}nes, Anthony and Gaudreault, Jonathan and Vignault, Louis-Philippe and Bernard, Fr{\'e}d{\'e}ric and Quimper, Claude-Guy},
  booktitle={2020 IEEE International Conference on Systems, Man, and Cybernetics (SMC)},
  pages={1451--1458},
  year={2020},
  organization={IEEE}
}

@article{kullman2021frvcpy,
  title={frvcpy: An open-source solver for the fixed route vehicle charging problem},
  author={Kullman, Nicholas D and Froger, Aurelien and Mendoza, Jorge E and Goodson, Justin C},
  journal={INFORMS Journal on Computing},
  volume={33},
  number={4},
  pages={1277--1283},
  year={2021},
  publisher={INFORMS}
}

@article{froger2019improved,
  title={Improved formulations and algorithmic components for the electric vehicle routing problem with nonlinear charging functions},
  author={Froger, Aur{\'e}lien and Mendoza, Jorge E and Jabali, Ola and Laporte, Gilbert},
  journal={Computers \& Operations Research},
  volume={104},
  pages={256--294},
  year={2019},
  publisher={Elsevier}
}

@article{dantzig1959truck,
  title={The truck dispatching problem},
  author={Dantzig, George B and Ramser, John H},
  journal={Management science},
  volume={6},
  number={1},
  pages={80--91},
  year={1959},
  publisher={Informs}
}

@article{konstantakopoulos2022vehicle,
  title={Vehicle routing problem and related algorithms for logistics distribution: a literature review and classification: GD Konstantakopoulos et al.},
  author={Konstantakopoulos, Grigorios D and Gayialis, Sotiris P and Kechagias, Evripidis P},
  journal={Operational research},
  volume={22},
  number={3},
  pages={2033--2062},
  year={2022},
  publisher={Springer}
}

@article{lenstra1981complexity,
  title={Complexity of vehicle routing and scheduling problems},
  author={Lenstra, Jan Karel and Kan, AHG Rinnooy},
  journal={Networks},
  volume={11},
  number={2},
  pages={221--227},
  year={1981},
  publisher={Wiley Online Library}
}

@article{tahami2020exact,
  title={Exact approaches for routing capacitated electric vehicles},
  author={Tahami, Hesamoddin and Rabadi, Ghaith and Haouari, Mohamed},
  journal={Transportation Research Part E: Logistics and Transportation Review},
  volume={144},
  pages={102126},
  year={2020},
  publisher={Elsevier}
}

@article{hien2023greedy,
  title={A greedy search based evolutionary algorithm for electric vehicle routing problem},
  author={Hien, Vu Quoc and Dao, Tran Cong and Binh, Huynh Thi Thanh},
  journal={Applied Intelligence},
  volume={53},
  number={3},
  pages={2908--2922},
  year={2023},
  publisher={Springer}
}

@article{rodriguez2024new,
  title={A new hyper-heuristic based on adaptive simulated annealing and reinforcement learning for the capacitated electric vehicle routing problem},
  author={Rodr{\'\i}guez-Esparza, Erick and Masegosa, Antonio D and Oliva, Diego and Onieva, Enrique},
  journal={Expert Systems with Applications},
  volume={252},
  pages={124197},
  year={2024},
  publisher={Elsevier}
}

@inproceedings{woller2020grasp,
  title={The grasp metaheuristic for the electric vehicle routing problem},
  author={Woller, David and Koz{\'a}k, Viktor and Kulich, Miroslav},
  booktitle={International conference on modelling and simulation for autonomous systems},
  pages={189--205},
  year={2020},
  organization={Springer}
}

@article{vidal2016split,
  title={Split algorithm in O (n) for the capacitated vehicle routing problem},
  author={Vidal, Thibaut},
  journal={Computers \& Operations Research},
  volume={69},
  pages={40--47},
  year={2016},
  publisher={Elsevier}
}

@inproceedings{durasevic2024automated,
  title={Automated Design of Routing Policies for the Dynamic Electric Vehicle Routing Problem with Genetic Programming.},
  author={Durasevic, Marko and Gala, Francisco Javier Gil},
  booktitle={IJCCI},
  pages={346--353},
  year={2024}
}

@article{lin2021deep,
  title={Deep reinforcement learning for the electric vehicle routing problem with time windows},
  author={Lin, Bo and Ghaddar, Bissan and Nathwani, Jatin},
  journal={IEEE Transactions on Intelligent Transportation Systems},
  volume={23},
  number={8},
  pages={11528--11538},
  year={2021},
  publisher={IEEE}
}

@article{basso2022dynamic,
  title={Dynamic stochastic electric vehicle routing with safe reinforcement learning},
  author={Basso, Rafael and Kulcs{\'a}r, Bal{\'a}zs and Sanchez-Diaz, Ivan and Qu, Xiaobo},
  journal={Transportation research part E: logistics and transportation review},
  volume={157},
  pages={102496},
  year={2022},
  publisher={Elsevier}
}

@article{tang2023energy,
  title={Energy-optimal routing for electric vehicles using deep reinforcement learning with transformer},
  author={Tang, Mengcheng and Zhuang, Weichao and Li, Bingbing and Liu, Haoji and Song, Ziyou and Yin, Guodong},
  journal={Applied Energy},
  volume={350},
  pages={121711},
  year={2023},
  publisher={Elsevier}
}

@article{montoya2016multi,
  title={A multi-space sampling heuristic for the green vehicle routing problem},
  author={Montoya, Alejandro and Gu{\'e}ret, Christelle and Mendoza, Jorge E and Villegas, Juan G},
  journal={Transportation Research Part C: Emerging Technologies},
  volume={70},
  pages={113--128},
  year={2016},
  publisher={Elsevier}
}

@article{lozano2013exact,
  title={On an exact method for the constrained shortest path problem},
  author={Lozano, Leonardo and Medaglia, Andr{\'e}s L},
  journal={Computers \& operations research},
  volume={40},
  number={1},
  pages={378--384},
  year={2013},
  publisher={Elsevier}
}

@article{roberti2016electric,
  title={The electric traveling salesman problem with time windows},
  author={Roberti, Roberto and Wen, Min},
  journal={Transportation Research Part E: Logistics and Transportation Review},
  volume={89},
  pages={32--52},
  year={2016},
  publisher={Elsevier}
}

@article{hiermann2016electric,
  title={The electric fleet size and mix vehicle routing problem with time windows and recharging stations},
  author={Hiermann, Gerhard and Puchinger, Jakob and Ropke, Stefan and Hartl, Richard F},
  journal={European Journal of Operational Research},
  volume={252},
  number={3},
  pages={995--1018},
  year={2016},
  publisher={Elsevier}
}

@article{schiffer2018adaptive,
  title={An adaptive large neighborhood search for the location-routing problem with intra-route facilities},
  author={Schiffer, Maximilian and Walther, Grit},
  journal={Transportation Science},
  volume={52},
  number={2},
  pages={331--352},
  year={2018},
  publisher={INFORMS}
}

@article{floyd1962algorithm,
  title={Algorithm 97: shortest path},
  author={Floyd, Robert W},
  journal={Communications of the ACM},
  volume={5},
  number={6},
  pages={345--345},
  year={1962},
  publisher={ACM New York, NY, USA}
}
